\documentclass[a4paper,reqno]{amsart}
\usepackage{amsmath,amssymb,mathrsfs}
\usepackage{txfonts,bm,pifont}
\usepackage{cite}
\usepackage{graphicx,xcolor}
\usepackage[dvipdfm,bookmarks=true,bookmarksnumbered=true,colorlinks=true]{hyperref}
\usepackage{comment}

\newtheorem{theorem}{Theorem}[section]
\newtheorem{proposition}[theorem]{Proposition}

\newtheorem{lemma}[theorem]{Lemma}
\newtheorem{remark}[theorem]{Remark}
\newtheorem{definition}[theorem]{Definition}

\numberwithin{equation}{section}
\newcommand{\wutilde}[1]{\vrule depth 0pt width 0pt%
{\raise0.8pt\hbox{$\smash{{\mathop{#1} \limits_{\displaystyle\widetilde{}}}}$}}}
\makeatletter
\long\def\@makecaption#1#2{
 \vskip 10pt
 \setbox\@tempboxa\hbox{#1. #2}
 \ifdim \wd\@tempboxa >\hsize #1. #2\par \else \hbox
to\hsize{\hfil\box\@tempboxa\hfil}
 \fi}
\makeatother
\newcommand{\set}[2]{\left\{\left. #1 ~\right|~ #2 \right\}}
\begin{document}
\title[Lattice equations arising from discrete Painlev\'e systems (I)]
{Lattice equations arising from discrete Painlev\'e systems (I): 
$(A_2+A_1)^{(1)}$ and $(A_1+A_1')^{(1)}$ cases}
\author{Nalini Joshi}
\author{Nobutaka Nakazono}
\address{School of Mathematics and Statistics, The University of Sydney, New South Wales 2006, Australia.}
\email{nobua.n1222@gmail.com}
\author{Yang Shi}
\begin{abstract}
We introduce the concept of  {\it $\omega$-lattice}, constructed from $\tau$ functions of Painlev\'e systems,
on which quad-equations of ABS type appear.
In particular, we consider the $A_5^{(1)}$- and $A_6^{(1)}$-surface $q$-Painlev\'e systems 
corresponding affine Weyl group symmetries are of $(A_2+A_1)^{(1)}$- and $(A_1+A_1)^{(1)}$-types, respectively.
\end{abstract}

\maketitle

\section{Introduction}
\subsection{Main result}
Although two important and widely used classifications of integrable discrete systems
have been known for more than a decade, no satisfactory relation between the two is yet understood.
The first is the ABS classification of integrable partial difference equations\cite{ABS2003:MR1962121,ABS2009:MR2503862,BollR2011:MR2846098,BollR2012:MR3010833,BollR:thesis}
while the second is Sakai's classification of integrable nonlinear ordinary difference equations\cite{SakaiH2001:MR1882403}.
Using geometry and symmetry groups of the equations as our main tool,
we present a new, general approach to connect the two classifications.

The framework we describe is based on the $\omega$-lattice, which is related to $\tau$-function theory.
While this framework is general, we explain its construction
for the $A_5^{(1)}$- and $A_6^{(1)}$-surface $q$-Painlev\'e systems 
and show how the ABS quad-equations appear.

The construction of the $\omega$-lattice is essential for knowledge about
how ABS-type equations can be reduced to a discrete Painlev\'e equation.
It provides not only the type of equation, but also the combinatorial structure of the lattice before reduction.
In \cite{JNS2014:MR3291391}, we showed how to use this information to find a reduction of equations 
on a $4$-dimensional hypercube (4D cube) but we did not provide details of the $\omega$-lattice construction.
Subsequently, in \cite{JNS:paper3} we provided a comprehensive method 
for constructing Lax pairs of the $A_5^{(1)}$-surface $q$-Painlev\'e equations.
The construction of the $\omega$-lattices for $A_5^{(1)}$- and $A_6^{(1)}$-surface $q$-Painlev\'e systems
provided in the present paper leads to the following main result.
\begin{theorem}\label{maintheorem}
All quad-equations appearing on the $\omega$-lattice for $A_5^{(1)}$-surface $q$-Painlev\'e system, 
defined by Equation \eqref{eqn:A5_omegafun},
and those for $A_6^{(1)}$-surface $q$-Painlev\'e system, 
defined by Equation \eqref{eqn:A6_omegafun},
are of ABS type.
\end{theorem}

\subsection{Background}
Discrete Painlev\'e equations and ABS equations have been studied from various viewpoints.
In particular, Sakai \cite{SakaiH2001:MR1882403} gave a classification of discrete Painlev\'e equations
based on the geometric structure of rational surfaces, and their corresponding affine Weyl symmetry group.
On the other hand, Adler, Bobenko and Suris \cite{ABS2003:MR1962121,ABS2009:MR2503862},
and later Boll \cite{BollR2011:MR2846098,BollR2012:MR3010833,BollR:thesis}, showed how to classify quad-equations
(partial difference equations on quadrilateral lattices) based on consistency of the equations on
$3$-dimensional cubes (see Section \ref{subsection:ABS}).
The resulting equations are called ABS equations.

Many types of periodic reductions from ABS equations to discrete Painlev\'e equations have been investigated
\cite{NP1991:MR1098879,GRSWC2005:MR2117991,JNS2014:MR3291391,FJN2008:MR2425981,HHJN2007:MR2303490,OrmerodCM2012:MR2997166,HHNS2015:MR3317164,OrmerodCM:2014arXiv1308.4233}.
It is well known that some discrete Painlev\'e equations
can be derived from ABS equations by periodic reductions with suitable choice of dependent variables.
However, 
(i) after applying a periodic reduction to an ABS equation,
we do not know which of discrete Painlev\'e equations appear;
(ii) discrete Painlev\'e equations obtained by periodic reductions often have insufficient number of parameters
(such an example appeared in \cite{GRSWC2005:MR2117991}, 
where the reduction which is given by Equations \eqref{eqn:intro_U_Omega}--\eqref{eqn:intro_qP3} for the special value $\lambda=1$ is discussed).
To solve the above-mentioned problems systematically, 
it is necessary to study the periodic reductions not only from the viewpoint of ABS equations
but also from that of Painlev\'e systems.
Unlike other investigations, which start with a quad-equation and obtain a discrete Painlev\'e equation,
we show how to obtain the reverse, by investigating underlying bilinear structure of $\tau$ functions
for the discrete Painlev\'e equation.

\subsection{ABS equation}\label{subsection:ABS}
In \cite{ABS2003:MR1962121,ABS2009:MR2503862,BollR2011:MR2846098,BollR2012:MR3010833,BollR:thesis}, 
Adler {\it et al.} classified polynomials in four variables into eleven types:
$Q4$, $Q3$, $Q2$, $Q1$, $H3$, $H2$, $H1$, $D4$, $D3$, $D2$, $D1$.
The first four types, the next three types and the last four types are collectively called $Q$-, $H^4$- and $H^6$-types, respectively.
The resulting polynomial $P$ satisfies the following properties.
\begin{description}
\item[(1) Linearity]
Polynomial $P$ is linear in each argument, i.e., it has the following form:
\begin{equation}
 P(x_1,x_2,x_3,x_4)=A_1x_1x_2x_3x_4+\cdots+A_{16},
\end{equation}
where coefficients $A_i$ are complex parameters.
\item[(2) 3D consistency and tetrahedron property]
There exist seven polynomials in four variables: $P^{(i)}$, $i=1,\dots,7$,
which satisfy the property {\bf(1)}
and a cube $C$ on whose six faces the following quad-equations are assigned
\begin{subequations}
\begin{align}
 &P(x_0,x_1,x_2,x_{12})=0,
 &&P^{(1)}(x_0,x_2,x_3,x_{23})=0,\\
 &P^{(2)}(x_0,x_3,x_1,x_{31})=0,
 &&P^{(3)}(x_3,x_{31},x_{23},x_{123})=0,\\
 &P^{(4)}(x_1,x_{12},x_{31},x_{123})=0,
 &&P^{(5)}(x_2,x_{23},x_{12},x_{123})=0,
\end{align}
\end{subequations}
where eight variables $x_i$ are on the vertices of the cube,
such that 
$x_{123}$ can be uniquely expressed by the four variables $x_i$, $i=0,1,2,3$,
({\it 3D consistency})
and the following relations hold ({\it tetrahedron property}):
\begin{equation}
 P^{(6)}(x_0,x_{12},x_{23},x_{31})=0,\quad
 P^{(7)}(x_1,x_2,x_3,x_{123})=0.
\end{equation}
\end{description} 

We here list some polynomials of ABS type as follows:
\begin{align*}
 Q1&:Q1(x_1,x_2,x_3,x_4;\alpha_1,\alpha_2;\epsilon)\\
 &\quad=\alpha_1(x_1x_2+x_3x_4)-\alpha_2(x_1x_4+x_2x_3)
 -(\alpha_1-\alpha_2)(x_1x_3+x_2x_4)+\epsilon\alpha_1\alpha_2(\alpha_1-\alpha_2),\\
 H3&:H3(x_1,x_2,x_3,x_4;\alpha_1,\alpha_2;\delta;\epsilon)\\
 &\quad=\alpha_1(x_1x_2+x_3x_4)-\alpha_2(x_1x_4+x_2x_3)
 +({\alpha_1}^2-{\alpha_2}^2)\left(\delta+\cfrac{\epsilon}{\alpha_1\alpha_2}\,x_2x_4\right),\\
 H1&:H1(x_1,x_2,x_3,x_4;\alpha_1,\alpha_2;\epsilon)
 =(x_1-x_3)(x_2-x_4)+(\alpha_2-\alpha_1)(1-\epsilon x_2x_4),\\
 D4&:D4(x_1,x_2,x_3,x_4;\delta_1,\delta_2,\delta_3)
 =x_1x_3+x_2x_4+\delta_1x_1x_4+\delta_2x_3x_4+\delta_3,
\end{align*}
where $\alpha_1,\alpha_2\in\mathbb{C}^\ast$ and $\epsilon,\delta,\delta_1,\delta_2,\delta_3\in\{0,1\}$.
It is well known that assigning a polynomial of ABS type to all faces of the integer lattice $\mathbb{Z}^2$,
we can obtain an integrable partial difference equation, e.g.
\begin{description}
\item[discrete Schwarzian KdV equation\cite{NC1995:MR1329559,NCWQ1984:MR763123}]
\begin{equation}\label{eqn:intro_DSKdV_1}
 Q1(U,\overline{U},\widehat{\overline{U}},\widehat{U};\alpha,\beta;0)=0
 ~\Leftrightarrow~
 \cfrac{(U-\overline{U})(\widehat{U}-\widehat{\overline{U}})}{(U-\widehat{U})(\overline{U}-\widehat{\overline{U}})}=\cfrac{\alpha}{\beta}\,;
\end{equation}
\item[lattice modified KdV equation\cite{NC1995:MR1329559,NQC1983:MR719638,ABS2003:MR1962121}]
\begin{equation}\label{eqn:intro_LMKdV_1}
 H3(U,\overline{U},-\widehat{\overline{U}},\widehat{U};\alpha,\beta;0;0)=0
 ~\Leftrightarrow~
 \cfrac{\widehat{\overline{U}}}{U}=\cfrac{\alpha\overline{U}-\beta\widehat{U}}{\alpha\widehat{U}-\beta\overline{U}}\,;
\end{equation}
\item[lattice potential KdV equation\cite{HirotaR1977:MR0460934,NC1995:MR1329559}]
\begin{equation}
 H1(U,\overline{U},\widehat{\overline{U}},\widehat{U};\alpha,\beta;0)=0
 ~\Leftrightarrow~
 (U-\widehat{\overline{U}})(\overline{U}-\widehat{U})=\alpha-\beta\,;
\end{equation}
\item[discrete version of Volterra-Kac-van Moerbeke equation\cite{NC1995:MR1329559}]
\begin{align}
 &D4(1-(\alpha^{-1}\beta-1)U,\widehat{U},\overline{U},-1+(\alpha^{-1}\beta-1)\widehat{\overline{U}};0,0,0)=0\notag\\
 &~\Leftrightarrow~
 \cfrac{\widehat{U}}{\overline{U}}=\cfrac{(\beta-\alpha)U-\alpha}{(\beta-\alpha)\widehat{\overline{U}}-\alpha}\,,
\end{align}
\end{description}
where 
\begin{equation}\label{eqn:intro_notation_1}
 U=U_{l,m},\quad
 \alpha=\alpha_l,\quad
 \beta=\beta_m,\quad
 \bar{ }:l\to l+1,\quad
 \hat{ }:m\to m+1,\quad
 l,m\in\mathbb{Z}.
\end{equation}
Throughout this paper, we often call a partial difference equation by the type of corresponding ABS polynomial.

The relations between ABS equations and discrete Painlev\'e equations on the level of equations have been intensively investigated,
but those on the level of underlying structure have not been clarified.
Here, we show an example of such a relation 
(the special case $\lambda=1$ is first obtained in \cite{GRSWC2005:MR2117991}).
By letting
\begin{equation}\label{eqn:intro_U_Omega}
 U_{l,m}=\lambda^l \Omega_{l,m},
\end{equation}
and applying the $(1,-2)$-periodic condition
\begin{equation}
 \Omega_{l+1,m-2}=\Omega_{l,m},
\end{equation}
which implies the condition on the parameters
\begin{equation}
 \cfrac{\overline{\alpha}}{~\alpha~}=\cfrac{~\beta~}{\widehat{\widehat{\beta}}}\,,
\end{equation}
Equation \eqref{eqn:intro_LMKdV_1} can be reduced to
\begin{equation}\label{eqn:intro_LMKdV_2}
 \cfrac{\widehat{\widehat{\widehat{\Omega}}}}{\Omega}
 =\cfrac{\widehat{\Omega}-\lambda\cfrac{\alpha}{\beta}\,\widehat{\widehat{\Omega}}}
 {\lambda\left(\lambda\widehat{\widehat{\Omega}}-\cfrac{\alpha}{\beta}\,\widehat{\Omega}\right)}\,,
\end{equation}
where
\begin{equation}
 \Omega=\Omega_{l,m}.
\end{equation}
Substituting
\begin{equation}
 f=\lambda \cfrac{~\widehat{\widehat{\Omega}}~}{\widehat{\Omega}}\,,\quad
 g=\lambda \cfrac{~\widehat{\Omega}~}{\Omega}\,,\quad
 t=-\cfrac{\alpha}{\beta}\,,\quad 
 a=\cfrac{~\beta~}{\widehat{\beta}}\,,\quad
 q=\cfrac{\overline{\alpha}}{~\alpha~}=\cfrac{~\beta~}{\widehat{\widehat{\beta}}}\,,
\end{equation}
in Equation \eqref{eqn:intro_LMKdV_2}, we obtain the $A_5^{(1)}$-surface $q$-Painlev\'e equation known as a $q$-discrete analogue of Painlev\'e III equation (denoted by $q$-P$_{\rm III}$)
\cite{KTGR2000:MR1789477,SakaiH2001:MR1882403}:
\begin{equation}\label{eqn:intro_qP3}
 \overline{g}=\cfrac{\lambda^2}{gf}\,\cfrac{1+tf}{t+f}\,,\quad
 \overline{f}=\cfrac{\lambda^2}{f\overline{g}}\,\cfrac{1+at\overline{g}}{at+\overline{g}}\,.
\end{equation}

\subsection{Plan of the paper}
This paper is organized as follows: 
in Section \ref{section:tau_f}, we introduce the $\tau$ functions of $A_5^{(1)}$-surface $q$-Painlev\'e systems,
which have the extended affine Weyl group symmetry of type $(A_2+A_1)^{(1)}$.
Moreover, we show that 
$q$-Painlev\'e equations can be derived from
a birational representation of the extended affine Weyl group of type $(A_2+A_1)^{(1)}$.
In Section \ref{section:omega},
we construct a lattice where quad-equations appear,
and then derive various quad-equations of ABS type, as relations on the lattice.
In Sections \ref{section:omegaA1A1}, 
we summarize the result for the case of $A_6^{(1)}$-surface $q$-Painlev\'e systems.
Some concluding remarks are given in Section \ref{ConcludingRemarks}.
\section{Construction of lattices from affine Weyl group $\widetilde{W}((A_2+A_1)^{(1)})$}\label{section:tau_f}
In this section, we describe $3$-dimensional structures constructed by using the symmetry groups of discrete Painlev\'e equations. While the groups themselves are well known, the novel perspective we focus on is the construction of $3$-dimensional lattices based on $\tau$ functions and $q$-Painlev\'e systems.

\subsection{The $\tau$-lattice}
We describe the action of the family of B\"acklund transformations of $q$-P$_{\rm III}$ \eqref{eqn:intro_qP3}
on six particular variables associated with this system\cite{TsudaT2006:MR2207047}. 
Iterating these variables under the affine Weyl group actions, 
we obtain a system of $\tau$ functions, which form a $\tau$-lattice. 

The transformation group $\widetilde{W}((A_2+A_1)^{(1)})$ 
has 7 generators $s_0$, $s_1$, $s_2$, $\pi$, $w_0$, $w_1$, $r$.
Below, we describe their actions on parameters: $a_0$, $a_1$, $a_2$, $c$, 
and on variables: $\tau_i$, $\bar{\tau}_i$, $i=0,1,2$.  
Actions on parameters are given by 
\begin{align*}
 &s_i:(a_i,a_{i+1},a_{i+2},c)\to({a_i}^{-1},a_ia_{i+1},a_ia_{i+2},c),
 &&\pi:(a_0,a_1,a_2,c)\to(a_1,a_2,a_0,c),\\
 &w_0:(a_0,a_1,a_2,c)\to(a_0,a_1,a_2,c^{-1}),
 &&w_1:(a_0,a_1,a_2,c)\to(a_0,a_1,a_2,q^{-2}c^{-1}),\\
 &r:(a_0,a_1,a_2,c)\to(a_0,a_1,a_2,q^{-1}c^{-1}),
\end{align*}
while its actions on variables are given by
\begin{equation*}
\begin{array}{lll}
 s_i(\tau_i)=
  \cfrac{u_i\tau_{i+1}\bar{\tau}_{i-1}+\bar{\tau}_{i+1}\tau_{i-1}}
   {{u_i}^{1/2}\bar{\tau}_i},
 &&s_i(\tau_j) = \tau_j\quad (i\neq j),\\
 s_i(\bar{\tau}_i)= 
  \cfrac{v_i\bar{\tau}_{i+1}\tau_{i-1}+\tau_{i+1}\bar{\tau}_{i-1}}
   {{v_i}^{1/2}\tau_i},
 &&s_i(\bar{\tau}_j) = \bar{\tau}_j\quad (i\neq j),\\
 \pi(\tau_i)= \tau_{i+1},
 &&\pi(\bar{\tau}_i) = \bar{\tau}_{i+1},\\
 w_0(\bar{\tau}_i)= 
  \cfrac{{a_{i+1}}^{1/3}(\bar{\tau}_i\tau_{i+1}\tau_{i+2}
   + u_{i-1}\tau_i\bar{\tau}_{i+1}\tau_{i+2}
   + {u_{i+1}}^{-1}\tau_i\tau_{i+1}\bar{\tau}_{i+2})}
  {{a_{i+2}}^{1/3}\bar{\tau}_{i+1}\bar{\tau}_{i+2}},
 &&w_0(\tau_i) = \tau_i,\\
 w_1(\tau_i)=
  \cfrac{{a_{i+1}}^{1/3}(\tau_i\bar{\tau}_{i+1}\bar{\tau}_{i+2}
   + v_{i-1}\bar{\tau}_i\tau_{i+1}\bar{\tau}_{i+2}
   + {v_{i+1}}^{-1}\bar{\tau}_i\bar{\tau}_{i+1}\tau_{i+2})}
  {{a_{i+2}}^{1/3}\tau_{i+1}\tau_{i+2}},
 &&w_1(\bar{\tau}_i) = \bar{\tau}_i,\\
 r(\tau_i)= \bar{\tau}_i,
 &&r(\bar{\tau}_i) = \tau_i,
\end{array}
\end{equation*}
where
\begin{equation}
 u_i = q^{-1/3}c^{-2/3}a_i,\quad
 v_i = q^{1/3}c^{2/3}a_i,\quad
 q=a_0a_1a_2,
\end{equation}
and $i,j\in\mathbb{Z}/3\mathbb{Z}$.
For each element $w\in\widetilde{W}((A_2+A_1)^{(1)})$ and function $F=F(a_i,c,\tau_j,\bar{\tau}_k)$, 
we use the notation $w.F$ to mean $w.F=F(w.a_i,w.c,w.\tau_j,w.\bar{\tau}_k)$, that is, 
$w$ acts on the arguments from the left. 
\begin{remark}
Notations in this paper are related to those in \cite{TsudaT2006:MR2207047}
by the following correspondence:
\begin{align*}
 &(s_0,s_1,s_2,\pi,w_0,w_1,r)
 \to(s_0,s_1,s_2,\pi^2,r_1,r_0,\pi^3),\\
 &(a_0,a_1,a_2,c)
 \to(a_0,a_1,a_2,q^{-1}b_0),\\
 &(\tau_0,\tau_1,\tau_2,\bar{\tau}_0,\bar{\tau}_1,\bar{\tau}_2)
 \to(\tau_3,\tau_1,\tau_5,\tau_6,\tau_4,\tau_2).
\end{align*}
We also note that in \cite{TsudaT2006:MR2207047}
each element $w\in\widetilde{W}((A_2+A_1)^{(1)})$ acts on the arguments from the right,
whereas in the present paper it acts from the left.
\end{remark}

The following proposition shows that 
$\widetilde{W}((A_2+A_1)^{(1)})$ gives a representation of an extended affine Weyl group of type $(A_2+A_1)^{(1)}$. 
\begin{proposition}[\cite{TsudaT2006:MR2207047}]
The group of transformations $\widetilde{W}((A_2+A_1)^{(1)})=\langle s_0,s_1,s_2,\pi, w_0,w_1,r\rangle$ 
forms the extended affine Weyl group of type $(A_2+A_1)^{(1)}$. 
Namely, the transformations satisfy the fundamental relations
\begin{subequations}
\begin{align}
 &{s_i}^2=(s_is_{i+1})^3=\pi^3=1,\quad
 \pi s_i = s_{i+1}\pi,\quad
 (i\in\mathbb{Z}/3\mathbb{Z}),\\
 &{w_0}^2={w_1}^2=r^2=1,\quad
 rw_0=w_1r,
\end{align}
\end{subequations}
and the action of $\widetilde{W}(A_2^{(1)})=\langle s_0,s_1,s_2,\pi\rangle$ and 
that of $\widetilde{W}(A_1^{(1)})=\langle w_0,w_1,r\rangle$ commute.
Note that the parameters $q$ and $c$ are invariant under the action of 
$\widetilde{W}((A_2+A_1)^{(1)})$ and $\widetilde{W}(A_2^{(1)})$, respectively.
\end{proposition}

To iterate each variable $\tau_i$, $\bar{\tau}_i$, we need the following translations $T_i$, $i=1,2,3,4$, defined by
\begin{equation}\label{eqn:def_translations}
 T_1=\pi s_2s_1,\quad
 T_2=\pi s_0s_2,\quad
 T_3=\pi s_1s_0,\quad
 T_4=rw_0.
\end{equation}
The actions of these on the parameters are given by
\begin{subequations}
\begin{align}
 &T_1:(a_0,a_1,a_2,c)\to(qa_0,q^{-1}a_1,a_2,c),\\
 &T_2:(a_0,a_1,a_2,c)\to(a_0,qa_1,q^{-1}a_2,c),\\
 &T_3:(a_0,a_1,a_2,c)\to(q^{-1}a_0,a_1,qa_2,c),\\
 &T_4:(a_0,a_1,a_2,c)\to(a_0,a_1,a_2,qc).
\end{align}
\end{subequations}
Note that $T_i$, $i =1,2,3,4$, commute with each other and $T_1T_2T_3=1$.
We define $\tau$ functions by
\begin{equation}
 \tau^{n,m}_{N}={T_1}^n{T_2}^m{T_4}^N(\tau_1),
\end{equation}
where $n,m,N\in\mathbb{Z}$ 
and the $\tau$-lattice is as shown in Figure \ref{fig:tau_lattice}.
We note that 
\begin{equation}
 \tau_0=\tau^{-1,0}_{0},\quad
 \tau_1=\tau^{0,0}_{0},\quad
 \tau_2=\tau^{0,1}_{0},\quad
 \bar{\tau}_0=\tau^{-1,0}_{1},\quad
 \bar{\tau}_1=\tau^{0,0}_{1},\quad
 \bar{\tau}_2=\tau^{0,1}_{1}.
\end{equation}
\begin{remark}
By definition, action of $\widetilde{W}((A_2+A_1)^{(1)})$
gives the relations of points on $\tau$-lattice (bilinear equations)
and any point of $\tau$-lattice (or $\tau$ function) is determined by six initial points:
$\tau_i$, $\bar{\tau}_i$, $i=0,1,2$.
\end{remark}

\begin{figure}[t]
\includegraphics[width=1\textwidth]{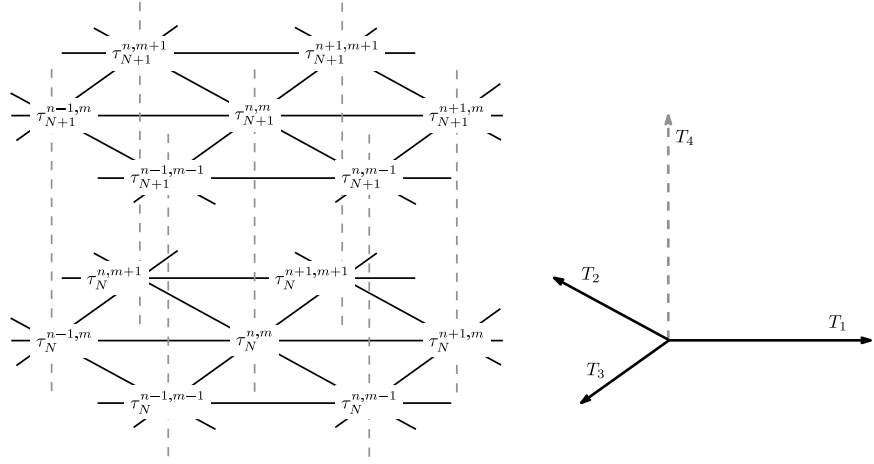}
\caption{Configuration of $\tau$ functions on the $\tau$-lattice. $\tau$ functions are defined on the intersections of four lines.}
\label{fig:tau_lattice}
\end{figure}

\subsection{The discrete Painlev\'e lattice}
In this section, we construct a $3$-dimensional lattice that relates ratios of $\tau$ functions. 
The ratios, defined in \eqref{eqn:f0f1f2}, turn out to satisfy a rich set of relations, 
which give rise not only to $q$-P$_{\rm III}$ \eqref{eqn:intro_qP3}, 
but also to $q$-P$_{\rm IV}$ \eqref{eqn:qp4} and $q$-P$_{\rm II}$ \eqref{eqn:qp2}
\cite{KNY2001:MR1876614,KNT2011:MR2773334}.  
We derive these $q$-Painlev\'e equations as relations on the $3$-dimensional lattice. 

The key starting point is the definition of the following ratios
\begin{equation}\label{eqn:f0f1f2}
 f_0=q^{1/3}c^{2/3}\cfrac{\bar{\tau}_1\tau_2}{\tau_1\bar{\tau}_2},\quad
 f_1=q^{1/3}c^{2/3}\cfrac{\bar{\tau}_2\tau_0}{\tau_2\bar{\tau}_0},\quad
 f_2=q^{1/3}c^{2/3}\cfrac{\bar{\tau}_0\tau_1}{\tau_0\bar{\tau}_1},
\end{equation}
where
\begin{equation}\label{eqn:relation_f012}
 f_0f_1f_2=qc^2.
\end{equation}
The action of $\widetilde{W}((A_2+A_1)^{(1)})$ on the variables $f_i$ is given by
\begin{align*}
 &s_i(f_{i-1})=f_{i-1}\cfrac{1+a_if_i}{a_i+f_i},\quad
 s_i(f_i)=f_i,\quad
 s_i(f_{i+1})=f_{i+1}\cfrac{a_i+f_i}{1+a_if_i},\quad
 \pi(f_i) = f_{i+1},\\
 &w_0(f_i)=\cfrac{a_ia_{i+1}(a_{i-1}a_i+a_{i-1}f_i+f_{i-1}f_i)}
  {f_{i-1}(a_ia_{i+1}+a_if_{i+1}+f_if_{i+1})},\\
 &w_1(f_i)=\cfrac{1+a_if_i+a_ia_{i+1}f_if_{i+1}}
  {a_ia_{i+1}f_{i+1}(1+a_{i-1}f_{i-1}+a_{i-1}a_if_{i-1}f_i)},\quad
  r(f_i)={f_i}^{-1},
\end{align*}
where $i\in\mathbb{Z}/3\mathbb{Z}$.
Define $f$-functions by
\begin{equation}\label{eqn:def_f0f1f2}
 f_{0,N}^{n,m}= {T_1}^n{T_2}^m{T_4}^N(f_0),\quad
 f_{1,N}^{n,m}= {T_1}^n{T_2}^m{T_4}^N(f_1),\quad
 f_{2,N}^{n,m}= {T_1}^n{T_2}^m{T_4}^N(f_2),
\end{equation}
where $n,m,N\in\mathbb{Z}$. 
These form the edges of a lattice, which we refer to as the $f$-lattice, shown in Figure \ref{fig:f_lattice}. 
This lattice is three-dimensional, with coordinate axes given by $n$, $m$, and $N$. 

\begin{figure}[t]
\includegraphics[width=1.0\textwidth]{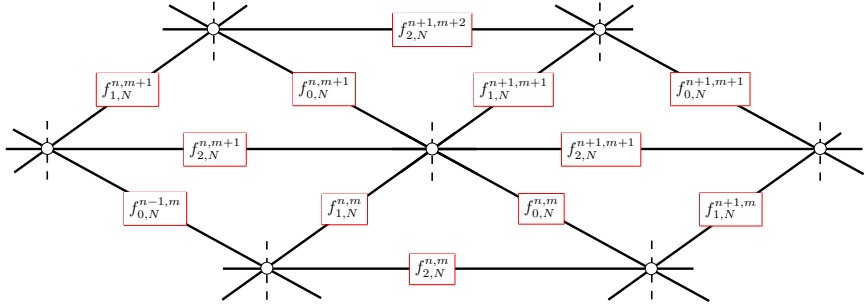}
\caption{Configuration of $f$-functions on the $f$-lattice. 
$f$-functions are defined on the edges of the triangle lattices.}
\label{fig:f_lattice}
\end{figure}

The relations in the $T_1$-direction on the lattice:
\begin{equation}
 T_1(f_1)=\cfrac{qc^2}{f_1f_0}\,\cfrac{1+a_0f_0}{a_0+f_0},\quad
 T_1(f_0)=\cfrac{qc^2}{f_0T_1(f_1)}\,\cfrac{1+a_0a_2T_1(f_1)}{a_0a_2+T_1(f_1)}
\end{equation}
lead to a system of first-order ordinary difference equations, 
which is equivalent to $q$-P$_{\rm III}$ \eqref{eqn:intro_qP3}:
\begin{equation}\label{eqn:qP3_lattice}
 f_{1,N}^{n+1,m} 
 =\cfrac{q^{2N+1}c^2}{f_{1,N}^{n,m}f_{0,N}^{n,m}}\,
  \cfrac{1+q^na_0f_{0,N}^{n,m}}{q^na_0+f_{0,N}^{n,m}},\quad
 f_{0,N}^{n+1,m}
 =\cfrac{q^{2N+1}c^2}{f_{0,N}^{n,m}f_{1,N}^{n+1,m}}\,
  \cfrac{1+q^{n-m}a_0a_2f_{1,N}^{n+1,m}}{q^{n-m}a_0a_2+f_{1,N}^{n+1,m}}.
\end{equation}
In a similar manner, in each of the $T_2$- and $T_3$-directions, we also obtain $q$-P$_{\rm III}$ \eqref{eqn:intro_qP3}.

In contrast, the action of $T_4$ on the variables $f_i$ can be expressed as
\begin{subequations}\label{eqn:action_T4}
\begin{align}
 &T_4(f_0)=a_0a_1 f_1\cfrac{1+a_2 f_2(a_0 f_0+1)}{1+a_0 f_0(a_1 f_1+1)},\\
 &T_4(f_1)=a_1a_2 f_2\cfrac{1+a_0 f_0(a_1 f_1+1)}{1+a_1 f_1(a_2 f_2+1)},\\
 &T_4(f_2)=a_2a_0 f_0\cfrac{1+a_1 f_1(a_2 f_2+1)}{1+a_2 f_2(a_0 f_0+1)},
\end{align}
\end{subequations}
or applying ${T_1}^n{T_2}^m{T_4}^N$ on System \eqref{eqn:action_T4} and using \eqref{eqn:def_f0f1f2},
we obtain
\begin{subequations}\label{eqn:qp4}
\begin{align}
 &f_{0,N+1}^{n,m}=q^m a_0a_1 f_{1,N}^{n,m}
 \cfrac{1+q^{-m}a_2 f_{2,N}^{n,m}(q^n a_0 f_{0,N}^{n,m}+1)}
  {1+q^na_0 f_{0,N}^{n,m}(q^{-n+m}a_1 f_{1,N}^{n,m}+1)},\\
 &f_{1,N+1}^{n,m}=q^{-n}a_1a_2 f_{2,N}^{n,m}
 \cfrac{1+q^na_0 f_{0,N}^{n,m}(q^{-n+m}a_1 f_{1,N}^{n,m}+1)}
  {1+q^{-n+m}a_1 f_{1,N}^{n,m}(q^{-m}a_2 f_{2,N}^{n,m}+1)},\\
 &f_{2,N+1}^{n,m}=q^{n-m}a_2a_0 f_{0,N}^{n,m}
 \cfrac{1+q^{-n+m}a_1 f_{1,N}^{n,m}(q^{-m}a_2 f_{2,N}^{n,m}+1)}
  {1+q^{-m}a_2 f_{2,N}^{n,m}(q^na_0 f_{0,N}^{n,m}+1)},
\end{align}
\end{subequations}
which is known as a $q$-discrete analogue of Painlev\'e IV equation (denoted by $q$-P$_{\rm IV}$)\cite{KNY2001:MR1876614}.

It is known that discrete dynamical systems of Painlev\'e type can be also obtained 
from elements of infinite order of (extended) affine Weyl groups which are not necessarily translations
\cite{KNT2011:MR2773334}.
We introduce the half-translation 
\begin{equation}
 R_1=\pi^2s_1
\end{equation} 
satisfying 
\begin{equation}
 {R_1}^2=T_1.
\end{equation}
Let
\begin{equation}
 f_N^M={R_1}^M{T_4}^N(f_0),
\end{equation}
where
\begin{equation}
 f_N^{2M-1}=f_{1,N}^{M,0},\quad
 f_N^{2M}=f_{0,N}^{M,0}.
\end{equation}
By considering the restricted $f$-lattice where $f_N^M$ are defined (see Figure \ref{fig:restricted_f_lattice}), 
System \eqref{eqn:qP3_lattice} becomes the following system:
\begin{equation}
 f_{1,N}^{n+1,0} = \cfrac{q^{2N+1}c^2}{f_{1,N}^{n,0}f_{0,N}^{n,0}}
  ~\cfrac{1+a_0q^nf_{0,N}^{n,0}}{a_0q^n+f_{0,N}^{n,0}},\quad
 f_{0,N}^{n+1,0} = \cfrac{q^{2N+1}c^2}{f_{0,N}^{n,0}f_{1,N}^{n+1,0}}
  ~\cfrac{1+a_2a_0q^n f_{1,N}^{n+1,0}}{a_2a_0q^n+f_{1,N}^{n+1,0}},
\end{equation}
which is equivalent to the following single equation:
\begin{equation}\label{eqn:R1_f0f1}
 f_N^{M+1}
 =\cfrac{q^{2N+1}c^2}{f_N^{M-1}f_N^M}\,\cfrac{1+{R_1}^M(a_0)f_N^M}{{R_1}^M(a_0)+f_N^M}.
\end{equation}
In addition, by assuming $a_2=q^{1/2}$, 
transformation $R_1$ becomes the translational motion in the parameter subspace:
\begin{equation}
 R_1:(a_0,a_1)\to (q^{1/2}a_0,q^{-1/2}a_1),
\end{equation}
then Equation \eqref{eqn:R1_f0f1} can be regarded as the single second-order ordinary difference equation:
\begin{equation}\label{eqn:qp2}
 f_N^{M+1}
 =\cfrac{q^{2N+1}c^2}{f_N^{M-1}f_N^M}\,\cfrac{1+a_0q^{M/2}f_N^M}{a_0q^{M/2}+f_N^M},
\end{equation}
which is known as a $q$-discrete analogue of Painlev\'e II equation (denoted by $q$-P$_{\rm II}$)\cite{RG1996:MR1399286}.
We note that the reduction from System \eqref{eqn:qP3_lattice} to Equation \eqref{eqn:qp2} is referred to as
a symmetrization or a projective reduction\cite{KNT2011:MR2773334,KN2015:MR3340349}.

\begin{figure}[t]
\includegraphics[width=1.0\textwidth]{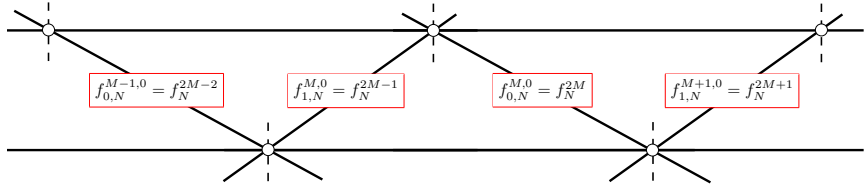}
\caption{Configuration of $f$-functions on the restricted $f$-lattice.}
\label{fig:restricted_f_lattice}
\end{figure}

\begin{remark}
By definition, action of $\widetilde{W}((A_2+A_1)^{(1)})$
gives the relations on each edge of the $f$-lattice. 
Since the variables $f_i$ satisfy Equation \eqref{eqn:relation_f012}, only two are independent.  
Therefore, any function associated with an edge on this lattice is determined by two initial edges. 
This is consistent with the observation that discrete Painlev\'e equations (which are second order ordinary)
are embedded in this lattice. 
\end{remark}

\section{Quad-equations of ABS type from the $\omega$-lattice for the $(A_2+A_1)^{(1)}$ case}
\label{section:omega}
In the previous section, we showed how to construct a $\tau$-lattice by starting with six initial variables
and how to obtain discrete Painlev\'e equations as relations on the $f$-lattice.
In this section, we show how to construct a lattice by starting with three variables 
and applying the action of the extended affine Weyl group to find their iterates. 
In the resulting $\omega$-lattice, we discover higher dimensional integrable partial difference equations, 
commonly known as quad-equations (because they relate vertices of quadrilaterals), 
that were classified by Adler {\it et al.}\cite{ABS2003:MR1962121,ABS2009:MR2503862,BollR2011:MR2846098,BollR2012:MR3010833,BollR:thesis}.  
\subsection{The $\omega$-lattice}
Let
\begin{equation}
 \kappa_0=\lambda^{\log{a_0}/\log{q}} k_0,\quad
 \kappa_1=\lambda^{\log{a_1}/\log{q}} k_1,\quad
 \kappa_2=\lambda^{\log{a_2}/\log{q}} k_2,
\end{equation}
where 
\begin{equation}
 \lambda=q^{1/2}c.
\end{equation}
Here, $k_i$ are arbitrary constants satisfying $k_0k_1k_2=1$.
The action of $\widetilde{W}((A_2+A_1)^{(1)})$ on the parameters $\kappa_i$ is given by
\begin{align*}
 &s_0:(\kappa_0,\kappa_1,\kappa_2)\to({\kappa_0}^{-1},\kappa_1\kappa_0,\kappa_2\kappa_0),\\
 &s_1:(\kappa_0,\kappa_1,\kappa_2)\to(\kappa_0\kappa_1,{\kappa_1}^{-1},\kappa_2\kappa_1),\\
 &s_2:(\kappa_0,\kappa_1,\kappa_2)\to(\kappa_0\kappa_2,\kappa_1\kappa_2,{\kappa_2}^{-1}),\\
 &\pi:(\kappa_0,\kappa_1,\kappa_2) \to(\kappa_1,\kappa_2,\kappa_0),\\
 &w_0:(\kappa_0,\kappa_1,\kappa_2)\to(a_0{\kappa_0}^{-1},a_1{\kappa_1}^{-1},a_2{\kappa_2}^{-1}),\\
 &w_1:(\kappa_0,\kappa_1,\kappa_2)\to({a_0}^{-1}{\kappa_0}^{-1},{a_1}^{-1}{\kappa_1}^{-1},{a_2}^{-1}{\kappa_2}^{-1}),\\
 &r:(\kappa_0,\kappa_1,\kappa_2)\to({\kappa_0}^{-1},{\kappa_1}^{-1},{\kappa_2}^{-1}).
\end{align*}
We note that $\kappa_i$ satisfy
\begin{equation}
 \kappa_0\kappa_1\kappa_2=\lambda.
\end{equation}
From definition \eqref{eqn:def_translations}, it follows that the actions of translations $T_i$, $i=1,2,3,4$, on parameters $\kappa_i$, $i=0,1,2$, are given by the following:
\begin{subequations}
\begin{align}
 &T_1:(\kappa_0,\kappa_1,\kappa_2)\to(\lambda\kappa_0,\lambda^{-1}\kappa_1,\kappa_2),\\
 &T_2:(\kappa_0,\kappa_1,\kappa_2)\to(\kappa_0,\lambda\kappa_1,\lambda^{-1}\kappa_2),\\
 &T_3:(\kappa_0,\kappa_1,\kappa_2)\to(\lambda^{-1}\kappa_0,\kappa_1,\lambda\kappa_2),\\
 &T_4:(\kappa_0,\kappa_1,\kappa_2)\to(a_0\kappa_0,a_1\kappa_1,a_2\kappa_2).
\end{align}
\end{subequations}

Now we are in a position to define the three initial variables 
\begin{equation}
 \omega_0=\cfrac{{\kappa_2}^{1/3}}{{\kappa_1}^{1/3}}\,\cfrac{\bar{\tau}_0}{\tau_0},\quad
 \omega_1=\cfrac{{\kappa_0}^{1/3}}{{\kappa_2}^{1/3}}\,\cfrac{\bar{\tau}_1}{\tau_1},\quad
 \omega_2=\cfrac{{\kappa_1}^{1/3}}{{\kappa_0}^{1/3}}\,\cfrac{\bar{\tau}_2}{\tau_2},
\end{equation}
whose iterates (constructed below) will provide us with the $\omega$-lattice.

The action of $\widetilde{W}((A_2+A_1)^{(1)})$  on these variables $\omega_i$ is given by the following lemma, which follows from the above definitions.
\begin{lemma}
The action of $\widetilde{W}((A_2+A_1)^{(1)})$ on variables $\omega_i$ is given by
\begin{align*}
 &s_i(\omega_i)
 =\omega_i\cfrac{a_i\lambda\omega_{i+1}+\kappa_i\omega_{i+2}}{\lambda\omega_{i+1}+a_i\kappa_i\omega_{i+2}},\quad
 s_i(\omega_{i+1})={\kappa_i}^{-1}\omega_{i+1},\quad
 s_i(\omega_{i+2})=\kappa_i\omega_{i+2},\\
 &\pi(\omega_i)=\omega_{i+1},\quad
 w_0(\omega_i)
 =\cfrac{a_{i+1}\kappa_i\kappa_{i+1}\omega_i+a_{i+1}a_{i+2}\omega_{i+1}+\kappa_i\lambda\omega_{i+2}}
 {a_{i+1}\kappa_i\kappa_{i+2}\omega_{i+1}\omega_{i+2}},\\
 &w_1(\omega_i)
 =\cfrac{a_{i+1}\kappa_i\kappa_{i+1}\omega_i}
  {a_{i+1}\kappa_i\kappa_{i+2}\omega_{i+1}\omega_{i+2}+a_{i+1}a_{i+2}\kappa_i\lambda\omega_i\omega_{i+2}+\omega_i\omega_{i+1}},\quad
 r(\omega_i)={\omega_i}^{-1},
\end{align*}
where $i\in\mathbb{Z}/3\mathbb{Z}$.
\end{lemma}

We define $\omega$-functions by
\begin{equation}\label{eqn:A5_omegafun}
 \omega_{l_1,l_2,l_3,l_4}={T_1}^{l_1}{T_2}^{l_2}{T_3}^{l_3}{T_4}^{l_4}(\omega_0),
\end{equation}
where $l_1,l_2,l_3,l_4\in\mathbb{Z}$ 
and the $\omega$-lattice is as shown in Figure \ref{fig:omega_lattice}.
We note that
\begin{equation}
 \omega_0=\omega_{0,0,0,0},\quad
 \omega_1={\kappa_2}^{-1}\omega_{1,0,0,0},\quad
 \omega_2=\kappa_1\omega_{1,1,0,0}.
\end{equation}

\begin{lemma}
Since for all $w\in \widetilde{W}((A_2+A_1)^{(1)})$,
\begin{equation}
 w(\omega_i)\in \mathcal{L}\quad (i=0,1,2),
\end{equation}
where $\mathcal{L}=\mathcal{K}(\omega_0,\omega_1,\omega_2)$ is the field of rational functions 
in $\omega_i$\,, $i=0,1,2$, with coefficient field $\mathcal{K}=\mathbb{C}(a_i,\kappa_i,\lambda)$,
every point on the $\omega$-lattice is determined by three initial points.
This implies that quad-equations appear as relations on the $\omega$-lattice.
Moreover, relations on the $f$-lattice can be expressed by those on the $\omega$-lattice
because of the following correspondence:
\begin{subequations}\label{eqn:relation_f_omega}
\begin{align}
 &f_0=\kappa_1\kappa_2\cfrac{\omega_1}{\omega_2},
 &&\left(\text{or }\ f_{0,l_4}^{l_1-l_3,l_2-l_3}=\cfrac{\omega_{l_1+1,l_2,l_3,l_4}}{\omega_{l_1+1,l_2+1,l_3,l_4}}\right),\\
 &f_1=\kappa_2\kappa_0\cfrac{\omega_2}{\omega_0},
 &&\left(\text{or }\ f_{1,l_4}^{l_1-l_3,l_2-l_3}=q^{l_4}\lambda \cfrac{\omega_{l_1+1,l_2+1,l_3,l_4}}{\omega_{l_1,l_2,l_3,l_4}}\right),\\
 &f_2=\kappa_0\kappa_1\cfrac{\omega_0}{\omega_1},
 &&\left(\text{or }\ f_{2,l_4}^{l_1-l_3,l_2-l_3}=q^{l_4}\lambda \cfrac{\omega_{l_1,l_2,l_3,l_4}}{\omega_{l_1+1,l_2,l_3,l_4}}\right).
\end{align}
\end{subequations}
\end{lemma}

\begin{figure}[t]
\hspace*{-5em}\includegraphics[width=1.3\textwidth]{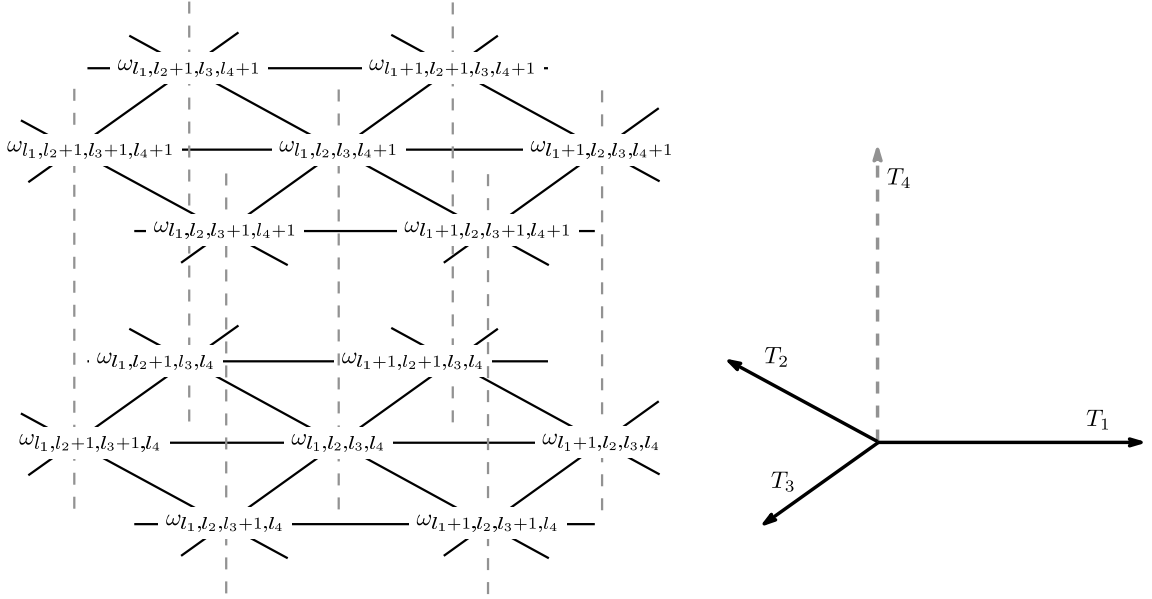}
\caption{Configuration of $\omega$-functions on the $\omega$-lattice. 
Note that each $\omega$-function is defined on the intersection of four lines.}
\label{fig:omega_lattice}
\end{figure}

We have constructed the $\omega$-lattice associated with $\widetilde{W}((A_2+A_1)^{(1)})$.
Henceforth, let us consider the quad-equations appearing on the $\omega$-lattice.
\begin{lemma}\label{lemma:quad-eqn_on_omega_lattice}
The following quad-equations:
\begin{subequations} \label{eqn:111H3}
\begin{align}
 &\cfrac{\omega_{l_1+1,l_2,l_3+1,l_4}}{\omega_{l_1,l_2,l_3,l_4}}
 =\cfrac{q^{l_1-l_3-1}a_0\omega_{l_1+1,l_2,l_3,l_4}-\omega_{l_1,l_2,l_3+1,l_4}}
 {q^{l_1-l_3-1}a_0\omega_{l_1,l_2,l_3+1,l_4}-\omega_{l_1+1,l_2,l_3,l_4}},
 \label{eqn:111H3_1}\\
 &\cfrac{\omega_{l_1+1,l_2+1,l_3,l_4}}{\omega_{l_1,l_2,l_3,l_4}}
 =\cfrac{q^{-l_1+l_2+l_4}\lambda a_1\omega_{l_1,l_2+1,l_3,l_4}-\omega_{l_1+1,l_2,l_3,l_4}}
 {q^{l_4}\lambda\left(q^{-l_1+l_2}a_1 \omega_{l_1+1,l_2,l_3,l_4}-q^{l_4}\lambda \omega_{l_1,l_2+1,l_3,l_4}\right)},
 \label{eqn:111H3_2}\\
 &\cfrac{\omega_{l_1,l_2+1,l_3+1,l_4}}{\omega_{l_1,l_2,l_3,l_4}}
 =\cfrac{q^{-l_2+l_3}a_2\omega_{l_1,l_2,l_3+1,l_4}-q^{l_4}\lambda \omega_{l_1,l_2+1,l_3,l_4}}
 {q^{l_4}\lambda\left(q^{-l_2+l_3+l_4}a_2\lambda \omega_{l_1,l_2+1,l_3,l_4}-\omega_{l_1,l_2,l_3+1,l_4}\right)},
 \label{eqn:111H3_3}
\end{align}
\end{subequations}
\begin{subequations}\label{eqn:111H6}
\begin{align}
 &\cfrac{\omega_{l_1+1,l_2,l_3,l_4+1}}{\omega_{l_1,l_2,l_3,l_4}}
  -\cfrac{\omega_{l_1,l_2,l_3,l_4+1}}{\omega_{l_1+1,l_2,l_3,l_4}}
  =\cfrac{q^{2l_4+1}\lambda^2-1}{q^{-l_1+l_2+l_4}\lambda a_1},
 \label{eqn:111H6_1}\\
 &\cfrac{\omega_{l_1,l_2+1,l_3,l_4+1}}{\omega_{l_1,l_2,l_3,l_4}}
  -\cfrac{1}{q^{2l_4+1}\lambda ^2}\,\cfrac{\omega_{l_1,l_2,l_3,l_4+1}}{\omega_{l_1,l_2+1,l_3,l_4}}
 =\cfrac{q^{2l_4+1}\lambda ^2-1}{q^{2l_4+1}\lambda ^2},
 \label{eqn:111H6_2}\\
 &\cfrac{\omega_{l_1,l_2,l_3+1,l_4+1}}{\omega_{l_1,l_2,l_3,l_4}}
  -\cfrac{\omega_{l_1,l_2,l_3,l_4+1}}{\omega_{l_1,l_2,l_3+1,l_4}}
 =\cfrac{a_2 (q^{2l_4+1}\lambda ^2-1)}{q^{l_2-l_3+l_4}\lambda},
 \label{eqn:111H6_3}
\end{align}
\end{subequations}
and an additional partial difference equation:
\begin{align}\label{eqn:T4_action}
 \omega_{l_1,l_2,l_3,l_4+1}
 =&\cfrac{(q^{2(-l_1+l_2)}{a_1}^2-1)\omega_{l_1+1,l_2,l_3,l_4}\omega_{l_1,l_2+1,l_3,l_4}}
  {q^{-l_1+l_2}a_1(q^{-l_1+l_2}a_1\omega_{l_1+1,l_2,l_3,l_4}-q^{l_4}\lambda \omega_{l_1,l_2+1,l_3,l_4})}\notag\\
 &+\cfrac{q^{-l_2+l_3}a_2
   (q^{-l_1+l_2+l_4}\lambda a_1\omega_{l_1,l_2+1,l_3,l_4}-\omega_{l_1+1,l_2,l_3,l_4})\omega_{l_1,l_2,l_3,l_4}}
 {q^{-l_1+l_2}a_1\omega_{l_1+1,l_2,l_3,l_4}-q^{l_4}\lambda \omega_{l_1,l_2+1,l_3,l_4}},
\end{align}
hold on the $\omega$-lattice.
We note that 
Equations \eqref{eqn:111H3} are $H3_{(\delta,\epsilon)=(0,0)}$-type equations,
Equations \eqref{eqn:111H6} are $D4_{(\delta_1,\delta_2,\delta_3)=(1,0,0)}$-type equations
and Equation \eqref{eqn:T4_action} is a $H3_{(\delta,\epsilon)=(0,1)}$-type equation.
\end{lemma}
\begin{proof}
First, we prove System \eqref{eqn:111H3}.
From the action of translations $T_1$, $T_2$ and $T_3$, it holds that
\begin{subequations}
\begin{align}
 &\cfrac{\omega_1}{\omega_2}
 =\cfrac{a_0\kappa_2\kappa_0 T_1(\omega_2) -\omega_0}
 {a_0\kappa_2\kappa_1 \omega_0 -\lambda\kappa_2 T_1(\omega_2) },\label{eqn:proof_lemma_H3_1}\\
 &\cfrac{\omega_2}{\omega_0}
 =\cfrac{a_1 \kappa_0\kappa_1 T_2(\omega_0) - \omega_1}
 {a_1 \kappa_0\kappa_2 \omega_1 -\lambda\kappa_0 T_2(\omega_0)},\label{eqn:proof_lemma_H3_2}\\
 &\cfrac{\omega_0}{\omega_1}
 =\cfrac{a_2 \kappa_1\kappa_2 T_3(\omega_1) -\omega_2}
 {a_2 \kappa_1\kappa_0 \omega_2 -\lambda\kappa_1 T_3(\omega_1)}.\label{eqn:proof_lemma_H3_3}
\end{align}
\end{subequations}
Applying ${T_1}^{l_1}{T_2}^{l_2}{T_3}^{l_3+1}{T_4}^{l_4}$, ${T_1}^{l_1}{T_2}^{l_2}{T_3}^{l_3}{T_4}^{l_4}$ 
and ${T_1}^{l_1}{T_2}^{l_2+1}{T_3}^{l_3+1}{T_4}^{l_4}$
on Equations \eqref{eqn:proof_lemma_H3_1}--\eqref{eqn:proof_lemma_H3_3},
we obtain Equations \eqref{eqn:111H3_1}--\eqref{eqn:111H3_3}, respectively.

We now consider the derivation of System \eqref{eqn:111H6}.
From the action of the translation $T_4$, it easily verified that
\begin{subequations}
\begin{align}
 &a_2\kappa_2\cfrac{T_4(\omega_1)}{\omega_0}
 -\cfrac{1}{\kappa_2}\,\cfrac{T_4(\omega_0)}{\omega_1}
  =\cfrac{q\lambda^2-1}{\lambda a_1},
 \label{eqn:proof_lemma_H6_1}\\
 &\cfrac{1}{a_1\kappa_1\kappa_2}\,\cfrac{T_4(\omega_2)}{\omega_1}
 -\cfrac{a_2\kappa_1\kappa_2}{q\lambda ^2}\,\cfrac{T_4(\omega_1)}{\omega_2}
 =\cfrac{q\lambda ^2-1}{q\lambda ^2},
 \label{eqn:proof_lemma_H6_2}\\
 &\kappa_1\cfrac{T_4(\omega_0)}{\omega_2}
 -\cfrac{1}{a_1\kappa_1}\,\cfrac{T_4(\omega_2)}{\omega_0}
 =\cfrac{a_2 (q\lambda ^2-1)}{q\lambda}.
 \label{eqn:proof_lemma_H6_3}
\end{align}
\end{subequations}
Applying ${T_1}^{l_1}{T_2}^{l_2}{T_3}^{l_3}{T_4}^{l_4}$, ${T_1}^{l_1-1}{T_2}^{l_2}{T_3}^{l_3}{T_4}^{l_4}$ 
and ${T_1}^{l_1-1}{T_2}^{l_2-1}{T_3}^{l_3}{T_4}^{l_4}$
on Equations \eqref{eqn:proof_lemma_H6_1}--\eqref{eqn:proof_lemma_H6_3},
we obtain Equations \eqref{eqn:111H6_1}--\eqref{eqn:111H6_3}, respectively.

Finally, we prove Equation \eqref{eqn:T4_action}.
By eliminating $\omega_2$ from Equation \eqref{eqn:proof_lemma_H3_2} and 
\begin{equation}
 T_4(\omega_0)
 =\cfrac{\omega_0\omega_1+a_1\kappa_0(\lambda a_2\omega_0+\kappa_2\omega_1)\omega_2}
 {a_1\kappa_0\kappa_1\omega_0},
\end{equation}
the following relation can be derived:
\begin{equation}\label{eqn:proof_lemma_T4_action}
 T_4(\omega_0)
 =\cfrac{({a_1}^2-1)\omega_1T_2(\omega_0)
   +a_1a_2(a_1\kappa_0\kappa_1T_2(\omega_0)-\omega_1)\omega_0}
 {a_1(a_1\omega_1-\kappa_0\kappa_1T_2(\omega_0))}.
\end{equation}
Applying ${T_1}^{l_1}{T_2}^{l_2}{T_3}^{l_3}{T_4}^{l_4}$ on Equation \eqref{eqn:proof_lemma_T4_action},
we obtain Equation \eqref{eqn:T4_action}.
This completes the proof.
\end{proof}

In \cite{JNS2014:MR3291391,JNS:paper3}, 
we showed the following proposition by using Lemma \ref{lemma:quad-eqn_on_omega_lattice}:
\begin{proposition}[\cite{JNS2014:MR3291391,JNS:paper3}]\label{prop:JNS_1}
The $\omega$-lattice can be obtained from 
an asymmetric 4D cube which has twelve $H3_{(\delta,\epsilon)=(0,0)}$-type equations 
and twelve $D4_{(\delta_1,\delta_2,\delta_3)=(1,0,0)}$-type equations
associated with each face.
\end{proposition}

It follows from Proposition \ref{prop:JNS_1} that the above quad-equations are the only ones
that relate four points on the $\omega$-variables.
Therefore, we have shown a part of Theorem \ref{maintheorem}.

\subsection{Relations to discrete Schwarzian KdV equation}
In this section, we show how to obtain the discrete Schwarzian KdV equation,
which is known as a discrete analogue of the Cauchy-Riemann relation (cross-ratio equation),
from the $\omega$-lattice.

In a recent work\cite{HKM2011:MR2788707}, Hay {\it et al.} showed that
by setting
\begin{equation}
 z_{l_1,l_2,l_3,l_4}={T_1}^{l_1}{T_2}^{l_2}{T_3}^{l_3}{T_4}^{l_4}(z),
\end{equation}
where
\begin{equation}
 z=c^{2\log{({a_1}^{-1}a_2)}/3\log{q}}\cfrac{{a_2}^{2/3}}{{a_1}^{2/3}}\,\cfrac{T_4(\bar{\tau}_0)}{\tau_0},
\end{equation}
one can obtain the discrete Schwarzian KdV equation:
\begin{equation}\label{eqn:111Q1_1}
 \cfrac{(z_{l_1,l_2,l_3,l_4}-z_{l_1+1,l_2,l_3,l_4})(z_{l_1,l_2,l_3+1,l_4}-z_{l_1+1,l_2,l_3+1,l_4})}
 {(z_{l_1,l_2,l_3,l_4}-z_{l_1,l_2,l_3+1,l_4})(z_{l_1+1,l_2,l_3,l_4}-z_{l_1+1,l_2,l_3+1,l_4})}
 =q^{2(l_1-l_3-1)}{a_0}^2.
\end{equation}
Equation \eqref{eqn:111Q1_1} can be found from the $\omega$-lattice
because of the following relation:
\begin{equation}\label{eqn:definition_z}
 z=T_4(\omega_0)\omega_0,\quad 
 (\text{or }\ z_{l_1,l_2,l_3,l_4}=\omega_{l_1,l_2,l_3,l_4+1}\omega_{l_1,l_2,l_3,l_4}).
\end{equation}
Furthermore, we can also obtain the following equations:
\begin{subequations}
\begin{align}
 &\cfrac{(z_{l_1,l_2,l_3,l_4}-q^{2l_4+1}\lambda^2z_{l_1,l_2+1,l_3,l_4})
   (z_{l_1+1,l_2,l_3,l_4}-q^{2l_4+1}\lambda^2z_{l_1+1,l_2+1,l_3,l_4})}
 {(z_{l_1,l_2,l_3,l_4}-z_{l_1+1,l_2,l_3,l_4})(z_{l_1,l_2+1,l_3,l_4}-z_{l_1+1,l_2+1,l_3,l_4})}
 =q^{-2l_1+2l_2+2l_4+1}\lambda^2{a_1}^2,
 \label{eqn:111Q1_2}\\
 &\cfrac{(z_{l_1,l_2,l_3,l_4}-z_{l_1,l_2,l_3+1,l_4})(z_{l_1,l_2+1,l_3,l_4}-z_{l_1,l_2+1,l_3+1,l_4})}
 {(z_{l_1,l_2,l_3,l_4}-q^{2l_4+1}\lambda^2z_{l_1,l_2+1,l_3,l_4})
  (z_{l_1,l_2,l_3+1,l_4}-q^{2l_4+1}\lambda^2z_{l_1,l_2+1,l_3+1,l_4})}
 =q^{-2l_2+2l_3-2l_4-1}\lambda^{-2}{a_2}^2,
 \label{eqn:111Q1_3}\\
 &\cfrac{(z_{l_1,l_2,l_3,l_4}-z_{l_1+1,l_2,l_3,l_4})(z_{l_1,l_2,l_3,l_4+1}-z_{l_1+1,l_2,l_3,l_4+1})}{z_{l_1,l_2,l_3,l_4}z_{l_1+1,l_2,l_3,l_4}}
 =\cfrac{(q^{2l_4+1}\lambda^2-1)(q^{2l_4+3}\lambda^2-1)}{q^{2(-l_1+l_2+l_4)+1}\lambda^2{a_1}^2},
 \label{eqn:A5_z_14}\\
 &\cfrac{(z_{l_1,l_2,l_3,l_4}-q^{2l_4+1}\lambda^2z_{l_1,l_2+1,l_3,l_4})(z_{l_1,l_2,l_3,l_4+1}-q^{2l_4+3}\lambda^2z_{l_1,l_2+1,l_3,l_4+1})}{z_{l_1,l_2,l_3,l_4}z_{l_1,l_2+1,l_3,l_4}}
 =(q^{2l_4+1}\lambda^2-1)(q^{2l_4+3}\lambda^2-1),
 \label{eqn:A5_z_24}\\
 &\cfrac{(z_{l_1,l_2,l_3,l_4}-z_{l_1,l_2,l_3+1,l_4})(z_{l_1,l_2,l_3,l_4+1}-z_{l_1,l_2,l_3+1,l_4+1})}{z_{l_1,l_2,l_3,l_4}z_{l_1,l_2,l_3+1,l_4}}
 =\cfrac{(q^{2l_4+1}\lambda^2-1)(q^{2l_4+3}\lambda^2-1)}{q^{2(l_2-l_3+l_4)+1}\lambda^2{a_2}^{-2}},
 \label{eqn:A5_z_34}
\end{align}
\end{subequations}
from the following relations:
\begin{subequations}
\begin{align}
 &\cfrac{(z-q\lambda^2 T_2(z))(T_1(z)-q\lambda^2 T_1T_2(z))}{(z-T_1(z))(T_2(z)-T_1T_2(z))}=q\lambda^2{a_1}^2,\\
 &\cfrac{(z-T_3(z))(T_2(z)-T_2T_3(z))}{(z-q\lambda^2 T_2(z))(T_3(z)-q\lambda^2T_2T_3(z))}=q^{-1}\lambda^{-2}{a_2}^2,\\
 &\cfrac{(z-T_1(z))(T_4(z)-T_1T_4(z))}{zT_1(z)}=\cfrac{(q\lambda^2-1)(q^3\lambda^2-1)}{q\lambda^2{a_1}^2},\\
 &\cfrac{(z-q\lambda^2T_2(z))(T_4(z)-q^3\lambda^2T_2T_4(z))}{zT_2(z)}=(q\lambda^2-1)(q^3\lambda^2-1),\\
 &\cfrac{(z-T_3(z))(T_4(z)-T_3T_4(z))}{zT_3(z)}=\cfrac{(q\lambda^2-1)(q^3\lambda^2-1)}{q\lambda^2{a_2}^{-2}}.
\end{align}
\end{subequations}
We note that 
Equations \eqref{eqn:111Q1_1}, \eqref{eqn:111Q1_2} and \eqref{eqn:111Q1_3} are $Q1_{\epsilon=0}$-type equations,
while Equations \eqref{eqn:A5_z_14}--\eqref{eqn:A5_z_34} are $H1_{\epsilon=0}$-type equations. 
\subsection{The restricted $\omega$-lattice}
In the case of the $f$-lattice, we showed that System \eqref{eqn:qP3_lattice}
can be rewritten as  the single equation \eqref{eqn:qp2} (or \eqref{eqn:R1_f0f1})
on the restricted lattice (projective reduction).
The concept of projective reduction applies not only for the $f$-lattice but also for the $\omega$-lattice.
Let 
\begin{equation}
 \omega_{l,l_4}={R_1}^l{T_4}^{l_4}(\omega_0),
\end{equation}
where
\begin{equation}
 \omega_{2l-1,l_4}=\cfrac{q^{l_4}\lambda}{{a_2}^{l_4}\kappa_2}\,\omega_{l,1,0,l_4},\quad
 \omega_{2l,l_4}=\omega_{l,0,0,l_4}.
\end{equation}
Considering the restricted $\omega$-lattice 
where $\omega_{l,l_4}$ are defined (see Figure \ref{fig:restricted_omega_lattice}),
Equations \eqref{eqn:111H3_1} and \eqref{eqn:111H3_2} can be rewritten as
\begin{subequations}
\begin{align}
 &\cfrac{\omega_{l+2,1,0,l_4}}{\omega_{l,0,0,l_4}}
 =\cfrac{\omega_{l+1,1,0,l_4}+q^{l}a_0\omega_{l+1,0,0,l_4}}
 {\omega_{l+1,0,0,l_4}+q^{l}a_0\omega_{l+1,1,0,l_4}},\label{eqn:111H3_1_2}\\
 &\cfrac{\omega_{l+1,0,0,l_4}}{\omega_{l,1,0,l_4}}
 =\cfrac{q^{l_4}\lambda \left(q^{-l}a_1 \omega_{l,0,0,l_4}+q^{l_4}\lambda \omega_{l+1,1,0,l_4}\right)}
 {q^{-l+l_4}\lambda a_1 \omega_{l+1,1,0,l_4}+\omega_{l,0,0,l_4}},
 \label{eqn:111H3_2_2}
\end{align}
\end{subequations}
respectively.
The system of equations \eqref{eqn:111H3_1_2} and \eqref{eqn:111H3_2_2} 
is expressed by the single equation
\begin{equation}\label{eqn:restricted_H3}
 \cfrac{\omega_{l+3,l_4}}{\omega_{l,l_4}}
  =\cfrac{q^{l_4}\lambda\left({R_1}^l{T_4}^{l_4}(\kappa_2)\omega_{l+1,l_4}+q^{l_4}\lambda{R_1}^l(a_0) \omega_{l+2,l_4}\right)}
  {{R_1}^l{T_4}^{l_4}(\kappa_2)
  \left(q^{l_4}\lambda \omega_{l+2,l_4}+{R_1}^l(a_0){R_1}^l{T_4}^{l_4}(\kappa_2) \omega_{l+1,l_4}\right)}.
\end{equation}
In a similar manner, Equation \eqref{eqn:111H6_1} becomes 
\begin{equation}\label{eqn:restricted_H6_1}
 \cfrac{\omega_{l+2,l_4+1}}{\omega_{l,l_4}}-\cfrac{\omega_{l,l_4+1}}{\omega_{l+2,l_4}}
 =\cfrac{q^{2l_4+1}\lambda^2-1}{q^{l_4}\lambda\, {R_1}^l(a_1)},
\end{equation}
and furthermore, Equations \eqref{eqn:111H6_2} and \eqref{eqn:111H6_3} 
can be expressed by 
\begin{equation}\label{eqn:restricted_H6_2}
 \cfrac{\omega_{l,l_4+1}}{\omega_{l+1,l_4}}
 -\cfrac{1}{{R_1}^{l-1}{T_4}^{l_4}(a_2{\kappa_2}^2)}\,\cfrac{\omega_{l+1,l_4+1}}{\omega_{l,l_4}}
 =\cfrac{q^{2l_4+1}\lambda^2-1}{q^{l_4}\lambda\, {R_1}^{l-1}{T_4}^{l_4}(a_2\kappa_2)},
\end{equation} 
on the restricted $\omega$-lattice.

\begin{figure}[t]
\begin{center}
\includegraphics[width=0.8\textwidth]{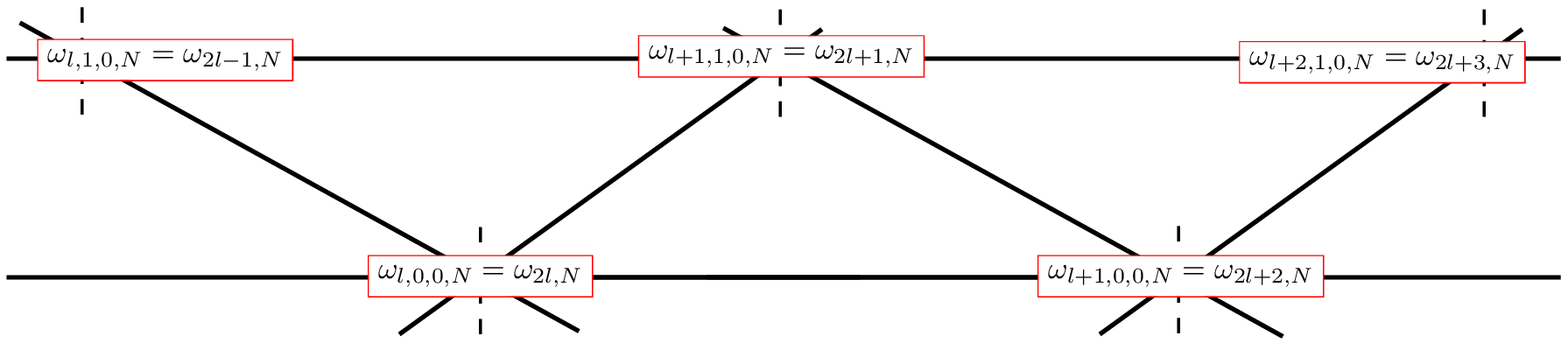}
\caption{Configuration of $\omega$-functions on the restricted $\omega$-lattice.}
\label{fig:restricted_omega_lattice}
\end{center}
\end{figure}

\begin{remark}[\cite{JNS:paper3}]
The restricted $\omega$-lattice can be also obtained from the asymmetric 4D cube.
\end{remark}

We note that on the restricted $\omega$-lattice, 
Equations \eqref{eqn:111Q1_1} and \eqref{eqn:111Q1_2},
Equation \eqref{eqn:A5_z_14} and
Equations \eqref{eqn:A5_z_24} and \eqref{eqn:A5_z_34} can be also rewritten as 
\begin{subequations}
\begin{align}
 &\cfrac{(\zeta_{l,l_4}-\zeta_{l+2,l_4})( \zeta_{l+1,l_4}- \zeta_{l+3,l_4})}
 {(q^{l_4+1/2}\lambda\, \zeta_{l,l_4}- \zeta_{l+1,l_4})(q^{l_4+1/2}\lambda\, \zeta_{l+2,l_4}- \zeta_{l+3,l_4})}
 =q^{-l_4-1/2}\lambda^{-1}{R_1}^{l}({a_0}^2),
 \label{eqn:2-1Q1}\\
 &\cfrac{(\zeta_{l,l_4}-\zeta_{l+2,l_4})(\zeta_{l,l_4+1}-\zeta_{l+2,l_4+1})}{\zeta_{l,l_4}\zeta_{l+2,l_4}}
 =\cfrac{(q^{2l_4+1}\lambda^2-1)(q^{2l_4+3}\lambda^2-1)}{q^{2l_4+1}\lambda^2{R_1}^{l}({a_1}^2a_2)},\\
 &\cfrac{(\zeta_{l,l_4}-q^{-(l_4+1)/2}\lambda^{-1}\zeta_{l+1,l_4})(\zeta_{l,l_4+1}-q^{-(l_4+3)/2}\lambda^{-1}\zeta_{l+1,l_4+1})}{\zeta_{l,l_4}\zeta_{l+1,l_4}}
 =\cfrac{(q^{2l_4+1}\lambda^2-1)(q^{2l_4+3}\lambda^2-1)}{q^{3l_4+5/2}\lambda^3{R_1}^{l}(a_0a_1)},
\end{align}
\end{subequations}
where
\begin{equation}\label{eqn:zeta_z}
 \zeta_{l,l_4}
 ={R_1}^{l}{T_4}^{l_4}({a_2}^{-1/2}{\kappa_2}^{-1})\,{R_1}^{l}{T_4}^{l_4}(z)
 ={R_1}^{l}{T_4}^{l_4}({a_2}^{-1/2}{\kappa_2}^{-1})\,\omega_{l,l_4}\omega_{l,l_4+1}.
\end{equation}

From periodic reduction of a partial difference equation,
we can obtain a quad-equation on the restricted $\omega$-lattice.
In fact, we obtain the following lemmas.

\begin{lemma}
Equation \eqref{eqn:restricted_H3} is equivalent to 
the $(1,-2)$-periodic reduction of $H3_{(\delta,\epsilon)=(0,0)}$ \eqref{eqn:intro_LMKdV_1}.
\end{lemma}
\begin{proof}
By setting
\begin{equation}
 \Omega_{l_1}^{l_4}={R_1}^{l}{T_4}^{l_4}(\lambda^{2/3}\kappa_1)\,\omega_{l,l_4},
\end{equation}
Equation \eqref{eqn:restricted_H3} can be rewritten as
\begin{equation}
 \cfrac{\Omega_{l+3}^{l_4}}{\Omega_l^{l_4}}
 =\cfrac{\Omega_{l+1}^{l_4}+q^{l_4}\lambda {R_1}^{l}(a_0)\Omega_{l+2}^{l_4}}
 {q^{l_4}\lambda\left(q^{l_4}\lambda\Omega_{l+2}^{l_4}+{R_1}^{l}(a_0)\Omega_{l+1}^{l_4}\right)},
\end{equation}
which is equivalent to Equation \eqref{eqn:intro_LMKdV_2} with the following correspondence:
\begin{equation}
 \Omega_{0,0}=\Omega_0^0,\quad
 \cfrac{\alpha_0}{\beta_0}=a_0,\quad
 \bar{}=T_1,\quad
 \hat{}=R_1.
\end{equation}
This completes the proof.
\end{proof}

\begin{lemma}
Equation \eqref{eqn:2-1Q1} can be obtained by a periodic reduction of $Q1_{\epsilon=0}$ \eqref{eqn:intro_DSKdV_1} and the reduction is defined by
\begin{equation}
 U_{l,m}=q^{-m/2}\lambda^{-m}\zeta_l^m,
\end{equation}
with the $(1,-2)$-periodic condition
\begin{equation}\label{eqn:Q1_zeta1-2}
 \zeta_{l+1}^{m-2}=\zeta_l^m.
\end{equation}
\end{lemma}
\begin{proof}
The periodic condition \eqref{eqn:Q1_zeta1-2} implies the condition on the parameters
\begin{equation}
 \cfrac{\overline{\alpha}}{~\alpha~}=\cfrac{~\beta~}{\widehat{\widehat{\beta}}}=q^2.
\end{equation}
Therefore, Equation \eqref{eqn:intro_DSKdV_1} can be reduced to
\begin{equation}\label{eqn:proof_2-1Q1}
 \cfrac{(\zeta-\widehat{\widehat{\zeta}})(\widehat{\zeta}-\widehat{\widehat{\widehat{\zeta}}})}
 {(q^{1/2}\lambda\,\zeta-\widehat{\zeta})(q^{1/2}\lambda\,\widehat{\widehat{\zeta}}-\widehat{\widehat{\widehat{\zeta}}})}
 =\cfrac{\alpha}{q^{1/2}\lambda\,\beta}\,,
\end{equation}
where
\begin{equation}
 \zeta=\zeta_l^m,\quad
 \alpha=\alpha_l,\quad
 \beta=\beta_m.
\end{equation}
Then, the statement holds since 
Equation \eqref{eqn:2-1Q1} is equivalent to Equation \eqref{eqn:proof_2-1Q1}
with the following correspondence:
\begin{equation}
 \zeta_0^0=\zeta_{0,0},\quad
 \cfrac{\alpha_0}{\beta_0}={a_0}^2,\quad
 \bar{}=T_1,\quad 
 \hat{}=R_1.
\end{equation}
\end{proof}

\section{$\omega$-lattice for the $(A_1+A_1')^{(1)}$ case}
\label{section:omegaA1A1}
In this section, we construct the $\omega$-lattice associated with the extended affine Weyl group of type $(A_1+A_1)^{(1)}$.
\subsection{The $\tau$-lattice, $f$-lattice and $\omega$-lattice}
In this section, we first construct the $\tau$ functions of $A_6^{(1)}$-surface $q$-Painlev\'e systems
with the extended affine Weyl group of type $(A_1+A_1)^{(1)}$.
Then, we construct the $\tau$-lattice, discrete Painlev\'e lattice ($f$-lattice) and $\omega$-lattice.
The actions \eqref{Yamada_tau} were first obtained by Yamada \cite{YamadaY:A1A1tau}.

The transformation group $\widetilde{W}((A_1+A_1')^{(1)})$ 
has 5 generators $s_0$, $s_1$, $w_0$, $w_1$, $\pi$.
Below, we describe their actions on parameters: $a_0$, $a_1$, $b$, 
and on variables: $\tau_i$, $i=-3,\dots, 1$.  
\begin{lemma}\label{lemma:tau_A6}
The action of $\widetilde{W}((A_1+A_1')^{(1)})$ on parameters are given by 
\begin{align*}
 &s_0:(a_0,a_1,b)
 \mapsto
 \left(\cfrac{1}{a_0}, {a_0}^2 a_1, \cfrac{b}{a_0}\right),
 &&s_1:(a_0,a_1,b)
 \mapsto
 \left(a_0 {a_1}^2, \cfrac{1}{a_1}, a_1 b\right),\\
 &w_0:(a_0,a_1,b)
 \mapsto
 \left(\cfrac{1}{a_0}, \cfrac{1}{a_1}, \cfrac{b}{a_0}\right),
 &&w_1:(a_0,a_1,b)
 \mapsto
 \left(\cfrac{1}{a_0}, \cfrac{1}{a_1}, \cfrac{b}{{a_0}^2 a_1}\right),\\
 &\pi:(a_0,a_1,b)
 \mapsto
 \left(\cfrac{1}{a_1}, \cfrac{1}{a_0}, \cfrac{b}{a_0 a_1}\right),
\end{align*}
while its actions on variables are given by
\begin{subequations}\label{Yamada_tau}
\begin{align}
 &s_0:(\tau_{-3},\tau_{-1},\tau_1)
 \mapsto
 \left(\cfrac{a_0 \tau_1 {\tau_{-2}}^2+\tau_{-1} \tau_0 \tau_{-2}+\tau_{-3} {\tau_0}^2}{a_0 \tau_{-1} \tau_1},
 \cfrac{a_0 {\tau_0}^2+b \tau_{-2} \tau_2}{a_0 \tau_1},
 \cfrac{b \tau_{-2} \tau_2+{\tau_0}^2}{\tau_{-1}}\right),\\
 &s_1:(\tau_{-2},\tau_0)
 \mapsto
 \left(\cfrac{a_0 a_1 {\tau_{-1}}^2+b \tau_{-3} \tau_1}{a_0 a_1 \tau_0},
 \cfrac{a_0 {\tau_{-1}}^2+b \tau_{-3} \tau_1}{a_0 \tau_{-2}}
 \right),\\
 &w_0:(\tau_{-3},\tau_{-2},\tau_{-1},\tau_1)
 \mapsto
 \left(\tau_3,\tau_2,\tau_1,\tau_{-1}\right),\\
 &w_1:(\tau_{-3},\tau_{-2},\tau_0,\tau_1)
 \mapsto
 \left(\tau_1,\tau_0,\tau_{-2},\tau_{-3}\right),\\
 &\pi:(\tau_{-3},\tau_{-2},\tau_{-1},\tau_0,\tau_1)
 \mapsto
 \left(\tau_2, \tau_1, \tau_0, \tau_{-1}, \tau_{-2}\right),
\end{align}
\end{subequations}
where
\begin{equation}
 \tau_2=\cfrac{a_0 \left(\tau_{-1} \tau_0+\tau_{-2} \tau_1\right)}{b \tau_{-3}},\quad
 \tau_3=\cfrac{\tau_0 \tau_1+\tau_{-1} \tau_2}{b \tau_{-2}}.
\end{equation}
For each element $w\in\widetilde{W}((A_1+A_1')^{(1)})$ and function $F=F(a_i,b,\tau_j)$, 
we use the notation $w.F$ to mean $w.F=F(w.a_i,w.b,w.\tau_j)$, that is, 
$w$ acts on the arguments from the left. 
\end{lemma}

The proof of Lemma \ref{lemma:tau_A6} is given in Appendix \ref{section:proof_tauA6}.
We note that the group of transformations $\widetilde{W}((A_1+A_1')^{(1)})=\langle s_0,s_1,w_0,w_1,\pi\rangle$ 
forms the extended affine Weyl group of type $(A_1+A_1)^{(1)}$. 
Namely, the transformations satisfy the fundamental relations
\begin{subequations}\label{eqns:A1A1_fundamental}
\begin{align}
 &{s_0}^2={s_1}^2=(s_0s_1)^\infty=1,\quad
 {w_0}^2={w_1}^2=(w_0w_1)^\infty=1,\\
 &\pi^2=1,\quad
 \pi s_0=s_1\pi,\quad 
 \pi w_0=w_1\pi,
\end{align}
\end{subequations}
and the action of $W(A_1^{(1)})=\langle s_0,s_1\rangle$ and that of $W(A_{1,|\beta|^2=14}^{(1)})=\langle w_0,w_1\rangle$ commute.
We note that
the relation $(ww')^\infty=1$ for transformations $w$ and $w'$ means that
there is no positive integer $N$ such that $(ww')^N=1$.

To iterate each variable $\tau_i$, we need the following translations $T_i$, $i=1,2,3$, defined by
\begin{equation}
 T_1=w_0 w_1,\quad
 T_2=\pi s_1 w_0,\quad
 T_3=\pi s_0 w_0.
\end{equation}
Note that $T_i$, $i =1,2,3$, commute with each other and $T_1T_2T_3=1$.
The actions of these on the parameters are given by
\begin{subequations}
\begin{align}
 &T_1:(a_0,a_1,b)\to(a_0,a_1,qb),\\
 &T_2:(a_0,a_1,b)\to(qa_0,q^{-1}a_1,b),\\
 &T_3:(a_0,a_1,b)\to(q^{-1}a_0,qa_1,q^{-1}b),
\end{align}
\end{subequations}
where the parameter 
\begin{equation}
 q=a_0a_1
\end{equation}
is invariant under the action of translations.
We define $\tau$ functions by
\begin{equation}
 \tau^n_N={T_1}^n{T_2}^N(\tau_{-3}),
\end{equation}
where $n,N\in\mathbb{Z}$ 
and the $\tau$-lattice is as shown in Figure \ref{fig:A1A1_tau_lattice}.
We note that 
\begin{equation}
 \tau_{-3}=\tau^0_0,\quad
 \tau_{-2}=\tau^1_1,\quad
 \tau_{-1}=\tau^1_0,\quad
 \tau_0=\tau^2_1,\quad
 \tau_1=\tau^2_0,\quad
 \tau_2=\tau^3_1,\quad
 \tau_3=\tau^3_0.
\end{equation}

\begin{figure}[t]
\begin{center}
\includegraphics[width=0.8\textwidth]{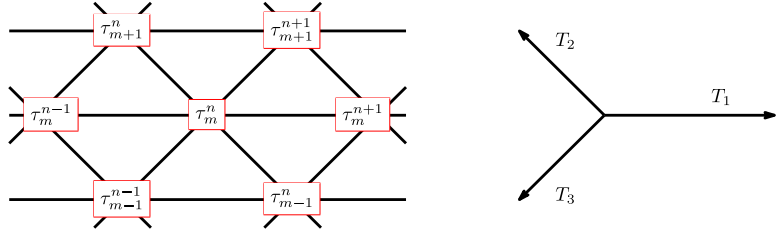}
\caption{Configuration of $\tau$ functions on the $\tau$-lattice.}
\label{fig:A1A1_tau_lattice}
\end{center}
\end{figure}

Next, we consider the $f$-lattice.
Let
\begin{equation}
 f_0=\cfrac{\tau _{-2} \tau _1}{\tau _{-1} \tau _0},\quad
 f_1=\cfrac{\tau _{-3} \tau _0}{\tau _{-2} \tau _{-1}},\quad
 f_2=\cfrac{{\tau_{-1}}^2}{\tau_{-3} \tau_1},
\end{equation}
where
\begin{equation}
 f_0f_1f_2=1.
\end{equation}
The action of $\widetilde{W}((A_1+A_1')^{(1)})$ on the variables $f_i$, $i=0,1,2$, is given by
\begin{align*}
 &s_0:(f_0,f_1,f_2)\mapsto
 \left(\cfrac{f_0 (a_0 f_0+a_0+f_1)}{f_0+f_1+1},\cfrac{f_1 (a_0 f_0+f_1+1)}{a_0 (f_0+f_1+1)},
 \cfrac{a_0 f_2(f_0+f_1+1)^2}{(a_0 f_0+a_0+f_1)(a_0 f_0+f_1+1)}\right),\\
 &s_1:(f_0,f_1)\mapsto
 \left(\cfrac{f_0 (a_0 a_1+b f_0 f_1)}{a_1 (a_0+b f_0 f_1)},\cfrac{a_1 f_1 (a_0+b f_0 f_1)}{a_0 a_1+b f_0 f_1}\right),\\
 &w_0:(f_0,f_1,f_2)\mapsto
 \left(\cfrac{a_0 (f_0+1)}{b f_0 f_1},\cfrac{a_0 f_0+a_0+b f_0 f_1}{a_0 b f_0 (f_0+1)},
 \cfrac{b^2 f_0}{f_2 (a_0 f_0+a_0+b f_0 f_1)}\right),\\
 &w_1:(f_0,f_1)\mapsto\left(f_1,f_0\right),\\
 &\pi:(f_1,f_2)\mapsto\left(\cfrac{a_0 (f_0+1)}{b f_0 f_1},\cfrac{b f_1}{a_0(f_0+1)}\right).
\end{align*}
Define $f$-functions by
\begin{equation}\label{eqn:def_f0f1}
 f_0^{l_1,l_2,l_3}={T_1}^{l_1}{T_2}^{l_2}{T_3}^{l_3}(f_0),\quad
 f_1^{l_1,l_2,l_3}={T_1}^{l_1}{T_2}^{l_2}{T_3}^{l_3}(f_1),\quad
 f_2^{l_1,l_2,l_3}={T_1}^{l_1}{T_2}^{l_2}{T_3}^{l_3}(f_2),
\end{equation}
where $n,N\in\mathbb{Z}$. 
These form the edges of a lattice, which we refer to as the $f$-lattice, shown in Figure \ref{fig:A1A1_f_lattice}. 
The relations in the $T_1$-direction on the lattice:
\begin{equation}
 T_1(f_1)=\cfrac{a_0(f_0+1)}{b f_0f_1},\quad
 T_1(f_0)=\cfrac{T_1(f_1)+1}{b T_1(f_1)f_0}
\end{equation}
lead to a system of first-order ordinary difference equations\cite{KTGR2000:MR1789477}:
\begin{equation}\label{eqn:A1A1_qp2_2}
 f_1^{l_1+1,l_2,l_3}=\cfrac{q^{l_2-l_3}a_0(f_0^{l_1,l_2,l_3}+1)}{q^{l_1-l_3}b f_0^{l_1,l_2,l_3} f_1^{l_1,l_2,l_3}},\quad
 f_0^{l_1+1,l_2,l_3}=\cfrac{f_1^{l_1+1,l_2,l_3}+1}{q^{l_1-l_3} b f_1^{l_1+1,l_2,l_3}f_0^{l_1,l_2,l_3}},
\end{equation}
or a single second-order ordinary difference equation\cite{RG1996:MR1399286,RGTT2001:MR1838017,SakaiH2001:MR1882403}:
\begin{equation}
 \left(f_0^{l_1+1,l_2,l_3}f_0^{l_1,l_2,l_3}-\cfrac{1}{q^{l_1-l_3}b}\right)
 \left(f_0^{l_1-1,l_2,l_3}f_0^{l_1,l_2,l_3}-\cfrac{1}{q^{l_1-l_3-1}b}\right)
 =\cfrac{q^{-l_1-l_2+2l_3}a_1}{b}\,\cfrac{f_0^{l_1,l_2,l_3}}{1+f_0^{l_1,l_2,l_3}},
\end{equation}
which is known as a $q$-discrete analogue of Painlev\'e II equation.
Moreover, from $T_2$- and $T_3$-directions we obtain the following systems of first-order ordinary difference equations:
\begin{align}
 &\begin{cases}
 f_2^{l_1,l_2+1,l_3}f_2^{l_1,l_2,l_3}
 =\cfrac{q^{l_1-l_3-1}b}{f_1^{l_1,l_2,l_3}(1+f_1^{l_1,l_2,l_3})},\\
 f_1^{l_1,l_2+1,l_3}f_1^{l_1,l_2,l_3}
 =\cfrac{q^{l_2-l_3}a_0(qf_2^{l_1,l_2+1,l_3}+q^{l_1-l_3}b)}{f_2^{l_1,l_2+1,l_3}(q^{l_2-l_3+1}a_0f_2^{l_1,l_2+1,l_3}+q^{l_1-l_3}b)},
 \end{cases}\\
 &\begin{cases}
 f_0^{l_1,l_2,l_3+1}f_0^{l_1,l_2,l_3}
 =\cfrac{qf_2^{l_1,l_2,l_3}+q^{l_1-l_2}a_1b}{f_2^{l_1,l_2,l_3}(qf_2^{l_1,l_2,l_3}+q^{l_1-l_3}b)},\\
 f_2^{l_1,l_2,l_3+1}f_2^{l_1,l_2,l_3}
 =\cfrac{q^{l_1-l_2-1}a_1b}{f_0^{l_1,l_2,l_3+1}(f_0^{l_1,l_2,l_3+1}+1)},
 \end{cases}
\end{align}
respectively.

\begin{figure}[t]
\begin{center}
\includegraphics[width=0.5\textwidth]{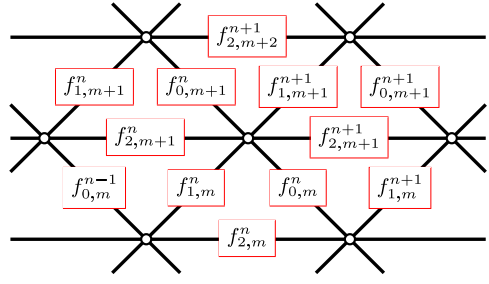}
\caption{Configuration of $f$-functions on the $f$-lattice.}
\label{fig:A1A1_f_lattice}
\end{center}
\end{figure}

Finally, we consider the $\omega$-lattice.
Letting
\begin{equation}
 \omega_0=\cfrac{\tau_{-1}}{\tau_{-3}},\quad
 \omega_1=\cfrac{\tau_1}{\tau_{-1}},\quad
 \omega_2=\cfrac{\tau_0}{\tau_{-2}},
\end{equation}
we obtain the action of $\widetilde{W}((A_1+A_1')^{(1)})$ on the variables $\omega_i$:
\begin{align*}
 &s_0:(\omega_0,\omega_1)
 \mapsto
 \left(
 \cfrac{a_0\omega_0 ({\omega_2}^2+\omega_0 \omega_2+\omega_0 \omega_1)}{a_0 \omega_0 \omega_1+{\omega_2}^2+\omega_0 \omega_2},
 \cfrac{\omega_1 (a_0 \omega_0 \omega_2+a_0 \omega_0 \omega_1+{\omega_2}^2)}{{\omega_2}^2+\omega_0 \omega_2+\omega_0 \omega_1}
 \right),\\
 &s_1:\omega_2
 \mapsto
 \cfrac{a_1\omega_2 (a_0 \omega_0+b \omega_1)}{a_0 a_1 \omega_0+b \omega_1},\\
 &w_0:(\omega_0,\omega_1,\omega_2)
 \mapsto
 \left(
 \cfrac{b^2 \omega_1}{a_0 \omega_0 \omega_1+a_0 \omega_0 \omega_2+b \omega_2 \omega_1},
 \cfrac{1}{\omega_1},
 \cfrac{b \omega_2}{a_0 (\omega_0 \omega_1+\omega_0 \omega_2)}
 \right),\\
 &w_1:(\omega_0,\omega_1,\omega_2)
 \mapsto
 \left(
 \cfrac{1}{\omega_1},
 \cfrac{1}{\omega_0},
 \cfrac{1}{\omega_2}
 \right),\\
 &\pi:(\omega_0,\omega_1,\omega_2)
 \mapsto
 \left(
 \cfrac{b \omega_2}{a_0 (\omega_0 \omega_1+\omega_0 \omega_2)},
 \cfrac{1}{\omega_2},
 \cfrac{1}{\omega_1}
 \right).
\end{align*}
We define $\omega$-functions by
\begin{equation}\label{eqn:A6_omegafun}
 \omega_{l_1,l_2,l_3}={T_1}^{l_1}{T_2}^{l_2}{T_3}^{l_3}(\omega_0),
\end{equation}
where $l_1,l_2,l_3\in\mathbb{Z}$ and the $\omega$-lattice is as shown in Figure \ref{fig:A1A1_omega_lattice}.
We note that
\begin{equation}
 \omega_0=\omega_{0,0,0},\quad
 \omega_1=\omega_{1,0,0},\quad
 \omega_2=\omega_{1,1,0}.
\end{equation}

\begin{lemma}
Since for all $w\in \widetilde{W}((A_1+A_1')^{(1)})$,
\begin{equation}
 w(\omega_i)\in \mathcal{L}\quad (i=0,1,2),
\end{equation}
where $\mathcal{L}=\mathcal{K}(\omega_0,\omega_1,\omega_2)$ is the field of rational functions 
in $\omega_i$\,, $i=0,1,2$, with coefficient field $\mathcal{K}=\mathbb{C}(a_0,a_1,b)$,
every point on the $\omega$-lattice is determined by three initial points.
This implies that quad-equations appear the relations on the $\omega$-lattice.
Moreover, relations on the $f$-lattice can be expressed by those on the $\omega$-lattice
because of the following correspondence:
\begin{subequations}
\begin{align}
 &f_0=\cfrac{\omega_1}{\omega_2},\qquad
 \left(\text{or }\ f_0^{l_1,l_2,l_3}=\cfrac{\omega_{l_1+1,l_2,l_3}}{\omega_{l_1+1,l_2+1,l_3}}\right),\\
 &f_1=\cfrac{\omega_2}{\omega_0},\qquad
 \left(\text{or }\ f_1^{l_1,l_2,l_3}=\cfrac{\omega_{l_1+1,l_2+1,l_3}}{\omega_{l_1,l_2,l_3}}\right),\\
 &f_2=\cfrac{\omega_0}{\omega_1},\qquad
 \left(\text{or }\ f_2^{l_1,l_2,l_3}=\cfrac{\omega_{l_1,l_2,l_3}}{\omega_{l_1+1,l_2,l_3}}\right).
\end{align}
\end{subequations}
\end{lemma}

\begin{figure}[t]
\begin{center}
\includegraphics[width=0.8\textwidth]{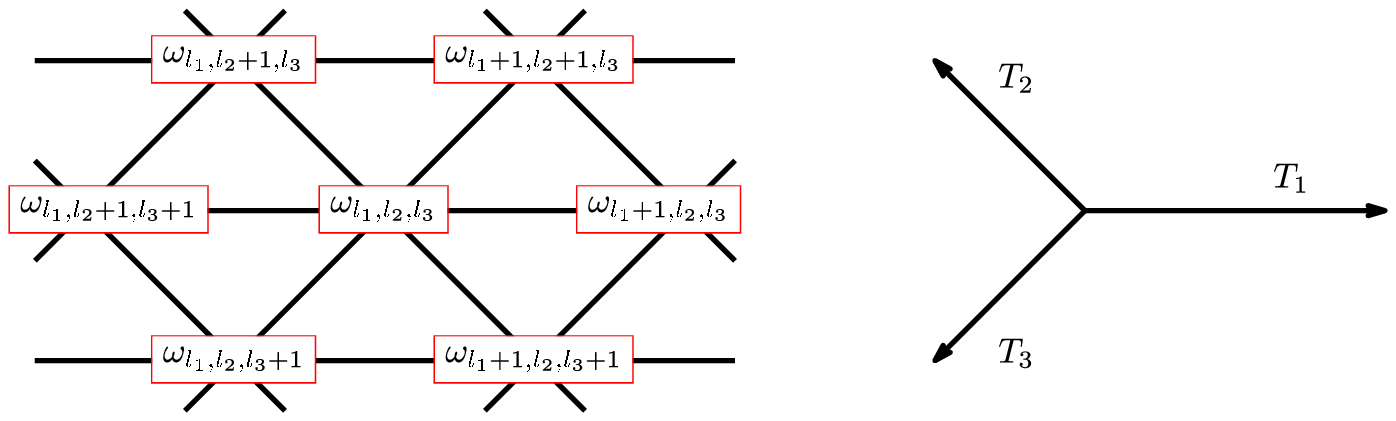}
\caption{Configuration of $\omega$-functions on the $\omega$-lattice.}
\label{fig:A1A1_omega_lattice}
\end{center}
\end{figure}

To show the relation between the $\omega$-lattice and a cube associated with ABS equations,
we derive relations \eqref{eqn:A1A1_H6H3} on the $\omega$-lattice.
\begin{lemma}\label{lemma:A1A1_quad-eqn_on_omega_lattice}
The following equations hold on the $\omega$-lattice:
\begin{subequations}\label{eqn:A1A1_H6H3}
\begin{align}
 &\cfrac{\omega_{l_1+1,l_2,l_3+1}}{\omega_{l_1,l_2,l_3}}
 -q^{l_1-l_2}\cfrac{b}{a_0}\,\cfrac{\omega_{l_1+1,l_2,l_3}}{\omega_{l_1,l_2,l_3+1}}=-1,
 \label{eqn:A1A1_H6_1}\\
 &\cfrac{\omega_{l_1+1,l_2+1,l_3}}{\omega_{l_1,l_2,l_3}}
 -q^{l_1-l_3-1}b\,\cfrac{\omega_{l_1+1,l_2,l_3}}{\omega_{l_1,l_2+1,l_3}}=-1,
 \label{eqn:A1A1_H6_2}\\
 &\cfrac{\omega_{l_1,l_2+1,l_3+1}}{\omega_{l_1,l_2,l_3}}
 =q^{l_1-l_2-1}\cfrac{b}{a_0}\,\cfrac{\omega_{l_1,l_2+1,l_3}-q^{l_2-l_3-1}a_0\omega_{l_1,l_2,l_3+1}}{\omega_{l_1,l_2,l_3+1}-\omega_{l_1,l_2+1,l_3}}.
 \label{eqn:A1A1_H3_1}
\end{align}
\end{subequations}
We note that 
Equations \eqref{eqn:A1A1_H6_1} and \eqref{eqn:A1A1_H6_2} are 
$D4_{(\delta_1,\delta_2,\delta_3)=(1,0,0)}$-type equations,
while Equation \eqref{eqn:A1A1_H3_1} is a $H3_{(\delta,\epsilon)=(0,0)}$-type equation.
\end{lemma}
\begin{proof}
From the action of translations $T_1$, $T_2$ and $T_3$, it follows that
\begin{subequations}
\begin{align}
 &\cfrac{\omega_1}{\omega_2}-\cfrac{b}{a_0}\,\cfrac{T_1(\omega_2)}{\omega_0}=-1,
 \label{eqn:A1A1_T1omega2}\\
 &\cfrac{\omega_2}{\omega_0}-q^{-1}b\,\cfrac{\omega_1}{T_2(\omega_0)}=-1,
 \label{eqn:A1A1_T2omega0}\\
 &\cfrac{\omega_0}{\omega_1}=\cfrac{b}{a_0}\,\cfrac{\omega_2-q^{-1}a_0T_3(\omega_1)}{T_3(\omega_1)-\omega_2}.
 \label{eqn:A1A1_T3omega1}
\end{align}
\end{subequations}
Applying ${T_1}^{l_1}{T_2}^{l_2}{T_3}^{l_3+1}$, ${T_1}^{l_1}{T_2}^{l_2}{T_3}^{l_3}$ and ${T_1}^{l_1}{T_2}^{l_2+1}{T_3}^{l_3+1}$
on Equations \eqref{eqn:A1A1_T1omega2}--\eqref{eqn:A1A1_T3omega1},
we obtain Equations \eqref{eqn:A1A1_H6_1}--\eqref{eqn:A1A1_H3_1}, respectively.
This completes the proof.
\end{proof}

In \cite{JNS:paper4}, we show that the following proposition follows from Lemma \ref{lemma:A1A1_quad-eqn_on_omega_lattice}.
\begin{proposition}[\cite{JNS:paper4}]\label{prop:JNS_2}
The $\omega$-lattice can be obtained from 
an asymmetric 3D cube which has two $H3_{(\delta,\epsilon)=(0,0)}$-type equations 
and four $D4_{(\delta_1,\delta_2,\delta_3)=(1,0,0)}$-type equations
associated with each face.
\end{proposition} 

It follows from Proposition \ref{prop:JNS_2} that the above quad-equations are the only ones
that relate four points on the $\omega$-variables.
This provides a part of Theorem \ref{maintheorem}.

In a similar manner as the case $(A_2+A_1)^{(1)}$, we can also obtain the discrete Schwarzian KdV equation
from the $\omega$-lattice.
Actually, letting
\begin{equation}
 z_{l_1,l_2,l_3}={T_1}^{l_1}{T_2}^{l_2}{T_3}^{l_3}(z),
\end{equation}
where
\begin{equation}
 z=\omega_0\omega_1,
\end{equation}
we obtain 
\begin{subequations}
\begin{align}
 &\cfrac{(z_{l_1,l_2,l_3}-q^{-l_1+l_3+1}b^{-1}z_{l_1,l_2+1,l_3})(z_{l_1+1,l_2,l_3}-q^{-l_1+l_3}b^{-1}z_{l_1+1,l_2+1,l_3})}
 {z_{l_1,l_2,l_3}z_{l_1,l_2+1,l_3}}
 =q^{2(-l_1+l_3)+1}b^{-2},
 \label{eqn:A6_z_12}\\
 &\cfrac{(z_{l_1,l_2,l_3}-q^{-l_1+l_2}a_0b^{-1}z_{l_1,l_2,l_3+1})(z_{l_1+1,l_2,l_3}-q^{-l_1+l_2-1}a_0b^{-1}z_{l_1+1,l_2,l_3+1})}
 {z_{l_1,l_2,l_3}z_{l_1,l_2,l_3+1}}
 =q^{2(-l_1+l_2)-1}{a_0}^2b^{-2},
 \label{eqn:A6_z_13}\\
 &\cfrac{(z_{l_1,l_2,l_3}-q^{-l_1+l_2}a_0b^{-1}z_{l_1,l_2,l_3+1})(z_{l_1,l_2+1,l_3}-q^{-l_1+l_2+1}a_0b^{-1}z_{l_1,l_2+1,l_3+1})}
 {(z_{l_1,l_2,l_3}-q^{-l_1+l_3+1}b^{-1}z_{l_1,l_2+1,l_3})(z_{l_1,l_2,l_3+1}-q^{-l_1+l_3+2}b^{-1}z_{l_1,l_2+1,l_3+1})}
 =q^{2(l_2-l_3)}{a_1}^{-2},
 \label{eqn:A6_z_23}
\end{align}
\end{subequations}
from the following relations:
\begin{subequations}
\begin{align}
 &\cfrac{(z-q b^{-1}T_2(z))(T_1(z)-b^{-1}T_1T_2(z))}{z T_2(z)}=q b^{-2},\\
 &\cfrac{(z-a_0b^{-1}T_3(z))(T_1(z)-q^{-1}a_0b^{-1}T_1T_3(z))}{z T_3(z)}=q^{-1}{a_0}^2b^{-2},\\
 &\cfrac{(z-a_0b^{-1}T_3(z))(T_2(z)-q a_0b^{-1}T_2T_3(z))}{(z-q b^{-1}T_2(z))(T_3(z)-q^2 b^{-1}T_2T_3(z))}={a_1}^{-2},
\end{align}
\end{subequations}
respectively.
We note that Equations \eqref{eqn:A6_z_12} and \eqref{eqn:A6_z_13} are $H1_{\epsilon=0}$-type equations,
while Equation \eqref{eqn:A6_z_23} is a $Q1_{\epsilon=0}$-type equation.
\subsection{The restricted cases}
In order to consider the restricted cases, we introduce the half-translation $R_1$ defined by
\begin{equation}
 R_1=\pi w_1,
\end{equation}
which satisfies
\begin{equation}
 {R_1}^2=T_1.
\end{equation}
The action of $R_1$ on the parameters is given by
\begin{equation*}
 R_1:(a_0,a_1,b)\mapsto (a_1,a_0,a_1b),
\end{equation*}
while its action on variables are given by
\begin{align*}
 &R_1:(\tau_{-3},\tau_{-2},\tau_{-1},\tau_0,\tau_1)\mapsto(\tau_{-2},\tau_{-1},\tau_0,\tau_1,\tau_2),\\
 &R_1:(f_0,f_1)\mapsto\left(\cfrac{a_0(f_0+1)}{b f_0 f_1},f_0\right),\\
 &R_1:(\omega_0,\omega_1,\omega_2)\mapsto\left(\omega_2,\cfrac{a_0 \omega_0(\omega_1+\omega_2)}{b \omega_2},\omega_1\right).
\end{align*}
The restricted functions are defined by
\begin{equation}
 \tau^{(l)}={R_1}^l(\tau_{-3}),\quad
 f^{(l)}={R_1}^l(f_0),\quad
 \omega^{(l)}={R_1}^l(\omega_0),
\end{equation}
where
\begin{subequations}
\begin{align}
 &\tau^{(2l)}=\tau_0^l,\quad
 \tau^{(2l+1)}=\tau_1^{l+1},\\
 &f^{(2l)}=f_0^{l,0,0},\quad
 f^{(2l+1)}=f_1^{l+1,0,0},\\
 &\omega^{(2l)}=\omega_{l,0,0},\quad
 \omega^{(2l+1)}=\omega_{l+1,1,0}.
\end{align}
\end{subequations}

System \eqref{eqn:A1A1_qp2_2} on the restricted $f$-lattice where $f^{(l)}$ are defined
can be rewritten as the following single equation
\begin{equation}\label{eqn:A1A1_R1}
 f^{(l+1)}f^{(l-1)}=\cfrac{{R_1}^l(a_0)(f^{(l)}+1)}{{R_1}^l(b)f^{(l)}}.
\end{equation}
We note that when 
\begin{equation}
 a_0=q^{1/2},
\end{equation}
transformation $R_1$ becomes the translational motion in the parameter subspace:
\begin{equation}
 R_1:b\mapsto q^{1/2}b,
\end{equation}
and then Equation \eqref{eqn:A1A1_R1} can be regarded as the single second-order ordinary difference equation:
\begin{equation}
 f^{(l+1)}f^{(l-1)}=\cfrac{f^{(l)}+1}{q^{(l-1)/2}bf^{(l)}},
\end{equation}
which is known as a $q$-discrete analogue of Painlev\'e I equation\cite{RG1996:MR1399286}.

Furthermore, Equations \eqref{eqn:A1A1_H6_1} and \eqref{eqn:A1A1_H6_2} 
on the restricted $\omega$-lattice where $\omega^{(l)}$ are defined
can be expressed by the following single equation:
\begin{equation}\label{eqn:restricted_H6}
 \cfrac{\omega^{(l+3)}}{\omega^{(l)}}
  -{R_1}^l\left(\cfrac{a_0}{b}\right)\cfrac{\omega^{(l+2)}}{\omega^{(l+1)}}
 ={R_1}^l\left(\cfrac{a_0}{b}\right).
\end{equation}
We note that Equations \eqref{eqn:A6_z_12} and \eqref{eqn:A6_z_13} on the restricted $\omega$-lattice
can be expressed by the following single equation:
\begin{equation}\label{eqn:restricted_H1}
 \cfrac{(z^{(l)}-{R_1}^l({a_0}^{-1}b)z^{(l+1)})(z^{(l+2)}-{R_1}^l(a_1b)z^{(l+3)})}{z^{(l)}z^{(l+1)}}=1,
\end{equation}
where
\begin{equation}
 z^{(l)}={R_1}^l(z).
\end{equation}

In a similar manner as the case $(A_2+A_1)^{(1)}$,
from periodic reductions of partial difference equations,
we can obtain the quad-equations on the restricted $\omega$-lattice as the following lemma.

\begin{lemma}
Equations \eqref{eqn:restricted_H6} and \eqref{eqn:restricted_H1}
can be respectively obtained by periodic reductions of $D4_{(\delta_1,\delta_2,\delta_3)=(1,0,0)}$- and $H1_{\epsilon=0}$-type equations:
\begin{align}
 &D4(-\alpha\beta^{-1}U,\widehat{\overline{U}},\overline{U},\widehat{U};1,0,0)=0\notag\\
 &\Leftrightarrow~
 \cfrac{\widehat{\overline{U}}}{~U~}-\cfrac{\alpha\,\overline{U}}{\beta\,\widehat{U}}
 =\cfrac{\alpha}{\beta}\,,
 \label{eqn:H6}\\
 &H1(U^{-1},\overline{\alpha}\beta\widehat{\overline{U}},\alpha^{-1}\beta^{-1}\widehat{U}^{-1},\overline{U};\alpha^{-2},\alpha^{-1}(\alpha^{-1}-\beta^{-1});0)=0\notag\\
 &\Leftrightarrow~
 \cfrac{(U-\alpha\beta\widehat{U})(\overline{U}-\overline{\alpha}\beta\widehat{\overline{U}})}{U\widehat{U}}=1,
 \label{ean:H1}
\end{align}
where we have used the notation \eqref{eqn:intro_notation_1}, with the $(1,-2)$-periodic condition
\begin{equation}\label{eqn:condition_H6}
U_{l+1,m-2}=U_{l,m}.
\end{equation}
\end{lemma}
\begin{proof}
Equation \eqref{eqn:H6} and periodic condition \eqref{eqn:condition_H6} imply the condition on the parameters
\begin{equation}
 \cfrac{\overline{\alpha}}{~\alpha~}=\cfrac{~\beta~}{\widehat{\widehat{\beta}}}=q^{-1}.
\end{equation}
Therefore, Equation \eqref{eqn:H6} can be reduced to
\begin{equation}\label{eqn:2-1H6}
 \cfrac{\widehat{\widehat{\widehat{U}}}}{~U~}
  -\cfrac{\alpha\,\widehat{\widehat{U}}}{\beta\,\widehat{U}}
 =\cfrac{\alpha}{\beta}.
\end{equation}
Then, the statement holds since 
Equation \eqref{eqn:2-1H6} is equivalent to Equation \eqref{eqn:restricted_H6}
with the following correspondence:
\begin{equation}
 U_{0,0}=\omega^{(0)},\quad
 \cfrac{\alpha_0}{\beta_0}=\cfrac{a_0}{b},\quad
 \bar{}=T_1,\quad 
 \hat{}=R_1.
\end{equation}

In a similar manner, we obtain Equation \eqref{eqn:restricted_H1} from 
the reduction of Equation \eqref{ean:H1} with the following correspondence:
\begin{equation}
 U_{0,0}=z^{(0)},\quad
 \alpha_0\beta_0=\cfrac{b}{a_0},\quad
 \cfrac{\overline{\alpha}}{\alpha}=q,\quad
 \bar{}=T_1,\quad 
 \hat{}=R_1.
\end{equation}
Therefore, we have completed the proof.
\end{proof}

\section{Concluding remarks}
\label{ConcludingRemarks}
In this paper, we constructed $\omega$-lattices associated with the extended affine Weyl groups of types $(A_2+A_1)^{(1)}$
and $(A_1+A_1)^{(1)}$.

More general $\omega$-lattices are possible. 
They share certain fundamental properties with the $\omega$-lattice constructed in Section \ref{section:omega}. 
In particular, all $f$-functions arise as rational combinations of $\omega$-functions and 
all $\omega$-functions in each connected component of an $\omega$-lattice 
are determined by three initial variables in that component. 
We will explore the general constructions of $\omega$-lattices in subsequent works. 
An interesting future project is to construct various $\omega$-lattices 
associated with Painlev\'e systems of other surface types in Sakai's classification\cite{SakaiH2001:MR1882403}.
\subsection*{Acknowledgment}
The authors would like to express their sincere thanks to Prof. Y. Yamada, Dr. J. Atkinson and Dr. P. Kassotakis
for fruitful discussions and valuable suggestions.
In particular, we are grateful to Prof. Y. Yamada for providing us the details of \cite{YamadaY:A1A1tau} before publication.

This research was supported by an Australian Laureate Fellowship \# FL120100094 and grant \# DP130100967 from the Australian Research Council.
\appendix
\section{Proof of Lemma \ref{lemma:tau_A6}}
\label{section:proof_tauA6}
In this section, we prove Lemma \ref{lemma:tau_A6}.

We consider the following eight base points:
\begin{subequations}\label{eqns:appendix_bps}
\begin{align}
 &P_1: (f,g)=(-1,0),\\
 &P_2: (f,g)=(0,-1),\\
 &P_3: (f,g)=(\infty,0),
 &&P_6: (f,g;fg)=(\infty,0;-q^{-1}b^{-1}),\\
 &P_4: (f,g)=(0,\infty),
 &&P_7: (f,g;fg)=(0,\infty;-a_0b^{-1}),\\
 &P_5: (f,g)=(\infty,\infty),
 &&P_8: (f,g;f/g)=(\infty,\infty;-{a_0}^{-1}),
\end{align}
\end{subequations}
where $a_0$, $b$, $q$ are complex parameters.
These base points are of the following $q$-difference equation:
\begin{equation}
 \bar{g}=\cfrac{a_0(f+1)}{b f g},\quad
 \bar{f}=\cfrac{\bar{g}+1}{b \bar{g} f},
\end{equation}
where $\,\bar{}\,$ means $t\mapsto qt$,
which is equivalent to System \eqref{eqn:A1A1_qp2_2}.
Let $\epsilon: X \to \mathbb{P}^1\times\mathbb{P}^1$ denotes blow up of $\mathbb{P}^1\times\mathbb{P}^1$ at the points \eqref{eqns:appendix_bps}.
The linear equivalence classes of the total transform of the coordinate lines $f$=constant and $g$=constant are denoted by $H_f$ and $H_g$, respectively. 
The Picard group of $X$, denoted by Pic$(X)$, is given by
\begin{equation}\label{eqn:pic_lat}
 {\rm Pic}(X)=\mathbb{Z}H_f\bigoplus\mathbb{Z}H_g\bigoplus_{i=1}^8\mathbb{Z}E_i,
\end{equation}
where $E_i=\epsilon^{-1}(P_i)$ is the total transform of the point of the $i$-th blow up. 
The intersection form $(|)$ is defined by 
\begin{equation}
 (H_f|H_g)=1,\quad
 (H_f|H_f)=(H_g|H_g)=(H_f|E_i)=(H_g|E_i)=0,\quad
 (E_i|E_j)=-\delta_{ij},
\end{equation}
where $1\le i\le 8$, $1\le j\le 8$ are integers. 
The anti-canonical divisor of $X$, denoted by $-K_X$, is uniquely decomposed into the prime divisors:
\begin{displaymath}
 \delta=-K_X=2H_f+2H_g-\sum_{i=1}^8E_i=\sum_{i=1}^7D_i,
\end{displaymath}
where
\begin{subequations}
\begin{align}
 &D_1=H_f-E_2-E_4,\quad
 D_2=E_4-E_7,\quad
 D_3=H_g-E_4-E_5,\quad
 D_4=E_5-E_8,\\
 &D_5=H_f-E_3-E_5,\quad
 D_6=E_3-E_6,\quad
 D_7=H_g-E_1-E_3.
\end{align}
\end{subequations}
We can show that the corresponding Cartan matrix
\begin{equation}
 (d_{ij})_{i,j=1}^7
 =\begin{pmatrix}
 2&-1&0&0&0&0&-1\\
 -1&2&-1&0&0&0&0\\
 0&-1&2&-1&0&0&0\\
 0&0&-1&2&-1&0&0\\
 0&0&0&-1&2&-1&0\\
 0&0&0&0&-1&2&-1\\
 -1&0&0&0&0&-1&2
 \end{pmatrix},
\end{equation}
where
\begin{equation}
 d_{ij}=\cfrac{2(D_i|D_j)}{(D_j|D_j)},
\end{equation}
and Dynkin diagram (see Figure \ref{fig:A6_dynkin}) are of type $A_6^{(1)}$.
Thus, we can set the root lattice as
\begin{equation}
 Q(A_6^{(1)})=\bigoplus_{i=1}^7\mathbb{Z}D_i,
\end{equation}
and identify the surface $X$ as being type $A_6^{(1)}$ in Sakai's list.

\begin{figure}[t]
\begin{tabular}{cc}
\begin{minipage}{0.5\hsize}
{\unitlength 0.1in%
\begin{picture}(27.1500,31.0000)(16.7000,-33.4000)%
%
\special{pn 20}%
\special{pa 2616 1000}%
\special{pa 3480 1000}%
\special{pa 4019 1675}%
\special{pa 3827 2518}%
\special{pa 3048 2893}%
\special{pa 2270 2518}%
\special{pa 2078 1675}%
\special{pa 2616 1000}%
\special{pa 3480 1000}%
\special{fp}%
%
\special{sh 0}%
\special{ia 2590 1000 90 90 0.0000000 6.2831853}%
\special{pn 8}%
\special{ar 2590 1000 90 90 0.0000000 6.2831853}%
\put(26.0000,-8.3000){\makebox(0,0){$D_1$}}%
%
\special{sh 0}%
\special{ia 3490 1000 90 90 0.0000000 6.2831853}%
\special{pn 8}%
\special{ar 3490 1000 90 90 0.0000000 6.2831853}%
%
\special{sh 0}%
\special{ia 4020 1670 90 90 0.0000000 6.2831853}%
\special{pn 8}%
\special{ar 4020 1670 90 90 0.0000000 6.2831853}%
%
\special{sh 0}%
\special{ia 2070 1670 90 90 0.0000000 6.2831853}%
\special{pn 8}%
\special{ar 2070 1670 90 90 0.0000000 6.2831853}%
\put(35.0000,-8.3000){\makebox(0,0){$D_2$}}%
%
\special{pn 8}%
\special{pa 3050 240}%
\special{pa 3050 3340}%
\special{dt 0.045}%
%
\special{pn 8}%
\special{pa 1705 880}%
\special{pa 4172 2757}%
\special{dt 0.045}%
%
\special{pn 8}%
\special{pa 4385 865}%
\special{pa 1940 2771}%
\special{dt 0.045}%
%
\special{sh 0}%
\special{ia 3050 2885 90 90 0.0000000 6.2831853}%
\special{pn 8}%
\special{ar 3050 2885 90 90 0.0000000 6.2831853}%
%
\special{sh 0}%
\special{ia 2280 2490 90 90 0.0000000 6.2831853}%
\special{pn 8}%
\special{ar 2280 2490 90 90 0.0000000 6.2831853}%
%
\special{sh 0}%
\special{ia 3830 2490 90 90 0.0000000 6.2831853}%
\special{pn 8}%
\special{ar 3830 2490 90 90 0.0000000 6.2831853}%
\put(40.7500,-14.9000){\makebox(0,0){$D_3$}}%
\put(40.1000,-23.8000){\makebox(0,0){$D_4$}}%
\put(32.1500,-30.3000){\makebox(0,0){$D_5$}}%
\put(21.1000,-23.8000){\makebox(0,0){$D_6$}}%
\put(20.7500,-14.9000){\makebox(0,0){$D_7$}}%
%
\special{pn 8}%
\special{ar 3050 750 300 300 3.7295953 5.6951827}%
%
\special{pn 8}%
\special{pa 3291 571}%
\special{pa 3300 584}%
\special{fp}%
\special{sh 1}%
\special{pa 3300 584}%
\special{pa 3278 518}%
\special{pa 3270 540}%
\special{pa 3246 541}%
\special{pa 3300 584}%
\special{fp}%
%
\special{pn 8}%
\special{pa 2809 571}%
\special{pa 2800 584}%
\special{fp}%
\special{sh 1}%
\special{pa 2800 584}%
\special{pa 2854 541}%
\special{pa 2830 540}%
\special{pa 2822 518}%
\special{pa 2800 584}%
\special{fp}%
%
\special{pn 8}%
\special{ar 4005 1140 300 300 4.6568905 0.3392926}%
%
\special{pn 8}%
\special{pa 4293 1225}%
\special{pa 4288 1240}%
\special{fp}%
\special{sh 1}%
\special{pa 4288 1240}%
\special{pa 4328 1183}%
\special{pa 4305 1189}%
\special{pa 4290 1170}%
\special{pa 4288 1240}%
\special{fp}%
%
\special{pn 8}%
\special{pa 4004 840}%
\special{pa 3988 840}%
\special{fp}%
\special{sh 1}%
\special{pa 3988 840}%
\special{pa 4055 860}%
\special{pa 4041 840}%
\special{pa 4055 820}%
\special{pa 3988 840}%
\special{fp}%
%
\special{pn 8}%
\special{ar 2080 1160 300 300 2.8514660 4.8208996}%
%
\special{pn 8}%
\special{pa 2097 860}%
\special{pa 2113 862}%
\special{fp}%
\special{sh 1}%
\special{pa 2113 862}%
\special{pa 2049 834}%
\special{pa 2060 855}%
\special{pa 2044 874}%
\special{pa 2113 862}%
\special{fp}%
%
\special{pn 8}%
\special{pa 1789 1231}%
\special{pa 1793 1246}%
\special{fp}%
\special{sh 1}%
\special{pa 1793 1246}%
\special{pa 1795 1176}%
\special{pa 1779 1194}%
\special{pa 1756 1187}%
\special{pa 1793 1246}%
\special{fp}%
\put(31.4000,-3.8500){\makebox(0,0){$\pi$}}%
\put(42.6000,-8.4500){\makebox(0,0){$w_0$}}%
\put(18.4500,-8.4000){\makebox(0,0){$w_1$}}%
\end{picture}}%
\end{minipage}&
\begin{minipage}{0.5\hsize}
\hspace{4em}
{\unitlength 0.1in%
\begin{picture}(14.2700,15.0000)(10.8000,-17.0000)%
%
\special{pn 13}%
\special{pa 1400 900}%
\special{pa 2400 900}%
\special{fp}%
%
\special{pn 13}%
\special{pa 1400 1000}%
\special{pa 2400 1000}%
\special{fp}%
%
\special{pn 8}%
\special{pa 1900 200}%
\special{pa 1900 1700}%
\special{dt 0.045}%
%
\special{sh 0}%
\special{ia 2427 948 80 80 0.0000000 6.2831853}%
\special{pn 8}%
\special{ar 2427 948 80 80 0.0000000 6.2831853}%
\put(24.2700,-7.7800){\makebox(0,0){$\alpha_1$}}%
%
\special{sh 0}%
\special{ia 1382 948 80 80 0.0000000 6.2831853}%
\special{pn 8}%
\special{ar 1382 948 80 80 0.0000000 6.2831853}%
\put(13.8200,-7.7800){\makebox(0,0){$\alpha_0$}}%
%
\special{pn 13}%
\special{pa 1400 1400}%
\special{pa 2400 1400}%
\special{fp}%
%
\special{pn 13}%
\special{pa 1400 1500}%
\special{pa 2400 1500}%
\special{fp}%
%
\special{sh 0}%
\special{ia 2427 1448 80 80 0.0000000 6.2831853}%
\special{pn 8}%
\special{ar 2427 1448 80 80 0.0000000 6.2831853}%
\put(24.2700,-12.7800){\makebox(0,0){$\beta_1$}}%
%
\special{sh 0}%
\special{ia 1382 1448 80 80 0.0000000 6.2831853}%
\special{pn 8}%
\special{ar 1382 1448 80 80 0.0000000 6.2831853}%
\put(13.8200,-12.7800){\makebox(0,0){$\beta_0$}}%
%
\special{pn 8}%
\special{ar 1900 630 200 200 3.4633432 5.9614348}%
%
\special{pn 8}%
\special{pa 2084 553}%
\special{pa 2090 567}%
\special{fp}%
\special{sh 1}%
\special{pa 2090 567}%
\special{pa 2082 498}%
\special{pa 2069 518}%
\special{pa 2045 514}%
\special{pa 2090 567}%
\special{fp}%
%
\special{pn 8}%
\special{pa 1716 553}%
\special{pa 1710 567}%
\special{fp}%
\special{sh 1}%
\special{pa 1710 567}%
\special{pa 1755 514}%
\special{pa 1731 518}%
\special{pa 1718 498}%
\special{pa 1710 567}%
\special{fp}%
\put(19.7500,-3.7500){\makebox(0,0){$\pi$}}%
\end{picture}}%
\end{minipage}
\end{tabular}
\caption{Dynkin diagrams for the root lattices. Left: $Q(A_6^{(1)})$, right: $Q(A_6^{(1)\bot})$.}
\label{fig:A6_dynkin}
\end{figure}
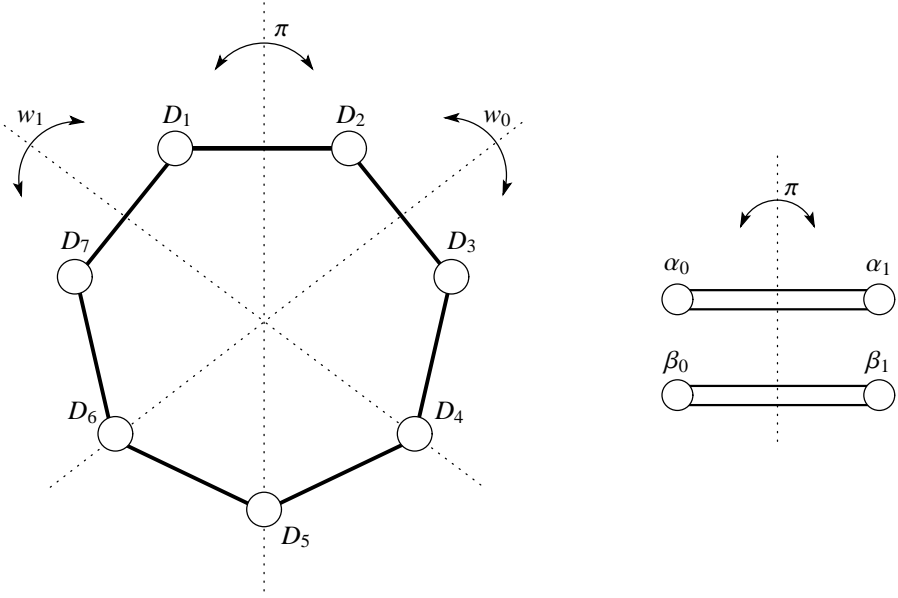

Moreover, we obtain the orthogonal root lattice 
\begin{equation}
 Q(A_6^{(1)\bot})
 =\mathbb{Z}\alpha_0\bigoplus\mathbb{Z}\alpha_1
 \bigoplus\mathbb{Z}\beta_0\bigoplus\mathbb{Z}\beta_1,
\end{equation}
where
\begin{subequations}
\begin{align}
 &\alpha_0=H_f+H_g-E_1-E_2-E_5-E_8,\\
 &\alpha_1=H_f+H_g-E_3-E_4-E_6-E_7,\\
 &\beta_0=3H_f+H_g-3E_1+E_2-2E_4-E_5-2E_7-E_8,\\
 &\beta_1=-H_f+H_g+2E_1-2E_2-E_3+E_4-E_6+E_7,\\
 &\delta=\alpha_0+\alpha_1=\beta_0+\beta_1,
\end{align}
\end{subequations}
by searching for elements of Pic$(X)$ that are orthogonal to all divisors $D_i$, $i=1,\dots,7$. 
The root lattice $Q(A_6^{(1)\bot})$ can be divided into the following two root lattice
\begin{equation}
 \mathbb{Z}\alpha_0\bigoplus\mathbb{Z}\alpha_1,\quad
 \mathbb{Z}\beta_0\bigoplus\mathbb{Z}\beta_1,
\end{equation}
since $\alpha_i$, $i=0,1$, and $\beta_i$, $i=0,1$, are orthogonal to each other: $(\alpha_i|\beta_j)=0$ where $i=0,1$ and $j=0,1$.
Moreover, their corresponding Cartan matrices
\begin{equation}
 (a_{ij})_{i,j=0}^1=(b_{ij})_{i,j=0}^1=\begin{pmatrix}2&-2\\-2&2\end{pmatrix},
\end{equation}
where
\begin{equation}
 a_{ij}=\cfrac{2(\alpha_i|\alpha_j)}{(\alpha_j|\alpha_j)},\quad
 b_{ij}=\cfrac{2(\beta_i|\beta_j)}{(\beta_j|\beta_j)},
\end{equation}
and Dynkin diagram (see Figure \ref{fig:A6_dynkin}) are of type $A_1^{(1)}$.
Therefore, we can set the root lattices as follows:
\begin{equation}
 Q(A_{1,|\alpha|^2=2}^{(1)})=\mathbb{Z}\alpha_0\bigoplus \mathbb{Z}\alpha_1,\quad
 Q(A_{1,|\beta|^2=14}^{(1)})=\mathbb{Z}\beta_0\bigoplus \mathbb{Z}\beta_1,
\end{equation}

Let us consider the Cremona isometries for this setting.
A Cremona isometry is defined by an automorphism of Pic$(X)$ which preserves 
\begin{description}
\item[(i)]
the intersection form on Pic$(X)$;
\item[(ii)]
the canonical divisor $K_X$;
\item[(iii)]
effectiveness of each effective divisor of Pic$(X)$.
\end{description}
It is well known that 
automorphisms of the Dynkin diagram corresponding to the divisors
and reflections for simple roots which orthogonal to all divisors
are Cremona isometries and form (extended) affine Weyl group\cite{SakaiH2001:MR1882403,DO1988:MR1007155,LooijengaE1981:MR632841}.
We define the reflections $s_i$, $i=0,1$, respectively across the hyperplane orthogonal to the root $\alpha_i$, $i=0,1$, by
\begin{equation}
 s_i(v)=v-\cfrac{2(v|\alpha_i)}{(\alpha_i|\alpha_i)}\,\alpha_i
\end{equation}
for all $v\in {\rm Pic}(X)$.
We can easily verify that the actions of $W(A_1^{(1)})=\langle s_0,s_1\rangle$ on Pic$(X)$ 
satisfy the fundamental relations of the affine Weyl group of type $A_1^{(1)}$:
\begin{equation}
 {s_0}^2={s_1}^2=(s_0s_1)^\infty=1.
\end{equation}
Note that the reflections corresponding to the roots $\beta_i$, $i=0,1$, cannot be constructed by this way 
since their self-intersection numbers are $-14$. 
We also define the group of the diagram automorphisms Aut$(A_6^{(1)})=\langle w_0,w_1,\pi\rangle$ by
{\allowdisplaybreaks
\begin{subequations}
\begin{align}
 &w_0:\begin{pmatrix}H_f\\H_g\\E_1\\E_2\\E_3\\E_4\\E_5\\E_6\\E_7\\E_8\end{pmatrix}
 \mapsto
 \begin{pmatrix}
 1&1&-1&0&0&-1&0&0&0&0\\
 2&1&-1&0&0&-1&-1&0&-1&0\\
 1&1&-1&0&0&-1&0&0&-1&0\\
 0&0&0&0&0&0&0&0&0&1\\
 0&0&0&0&1&0&0&0&0&0\\
 1&1&-1&0&0&-1&-1&0&0&0\\
 1&0&0&0&0&-1&0&0&0&0\\
 0&0&0&0&0&0&0&1&0&0\\
 1&0&-1&0&0&0&0&0&0&0\\
 0&0&0&1&0&0&0&0&0&0
 \end{pmatrix}
 \begin{pmatrix}H_f\\H_g\\E_1\\E_2\\E_3\\E_4\\E_5\\E_6\\E_7\\E_8\end{pmatrix},\\
 &w_1:\begin{pmatrix}H_f\\H_g\\E_1\\E_2\\E_3\\E_4\\E_5\\E_6\\E_7\\E_8\end{pmatrix}
 \mapsto
 \begin{pmatrix}
 0&1&0&0&0&0&0&0&0&0\\
 1&0&0&0&0&0&0&0&0&0\\
 0&0&0&1&0&0&0&0&0&0\\
 0&0&1&0&0&0&0&0&0&0\\
 0&0&0&0&0&1&0&0&0&0\\
 0&0&0&0&1&0&0&0&0&0\\
 0&0&0&0&0&0&1&0&0&0\\
 0&0&0&0&0&0&0&0&1&0\\
 0&0&0&0&0&0&0&1&0&0\\
 0&0&0&0&0&0&0&0&0&1
 \end{pmatrix}
 \begin{pmatrix}H_f\\H_g\\E_1\\E_2\\E_3\\E_4\\E_5\\E_6\\E_7\\E_8\end{pmatrix},\\
 &\pi:\begin{pmatrix}H_f\\H_g\\E_1\\E_2\\E_3\\E_4\\E_5\\E_6\\E_7\\E_8\end{pmatrix}
 \mapsto
 \begin{pmatrix}
 1&0&0&0&0&0&0&0&0&0\\
 1&1&-1&0&0&-1&0&0&0&0\\
 1&0&-1&0&0&0&0&0&0&0\\
 0&0&0&0&0&0&0&0&1&0\\
 0&0&0&0&0&0&1&0&0&0\\
 1&0&0&0&0&-1&0&0&0&0\\
 0&0&0&0&1&0&0&0&0&0\\
 0&0&0&0&0&0&0&0&0&1\\
 0&0&0&1&0&0&0&0&0&0\\
 0&0&0&0&0&0&0&1&0&0
 \end{pmatrix}
 \begin{pmatrix}H_f\\H_g\\E_1\\E_2\\E_3\\E_4\\E_5\\E_6\\E_7\\E_8\end{pmatrix},
\end{align}
\end{subequations}
}
where their actions on the divisors are given by
\begin{subequations}
\begin{align}
 w_0&:(D_1,D_2,D_3,D_4,D_5,D_6,D_7)\mapsto(D_4,D_3,D_2,D_1,D_7,D_6,D_5),\\
 w_1&:(D_1,D_2,D_3,D_4,D_5,D_6,D_7)\mapsto(D_7,D_6,D_5,D_4,D_3,D_2,D_1),\\
 \pi&:(D_1,D_2,D_3,D_4,D_5,D_6,D_7)\mapsto(D_2,D_1,D_7,D_6,D_5,D_4,D_3).
\end{align}
\end{subequations}

Since transformations $w_i$, $i=0,1$, respectively correspond to the reflections of the roots $\beta_i$, $i=0,1$, as follows:
\begin{subequations}
\begin{align}
 &w_0:(\alpha_0,\alpha_1,\beta_0,\beta_1)
 \mapsto(\alpha_0,\alpha_1,-\beta_0,\beta_1+2\beta_0),\\
 &w_1:(\alpha_0,\alpha_1,\beta_0,\beta_1)
 \mapsto(\alpha_0,\alpha_1,\beta_0+2\beta_1,-\beta_1),
\end{align}
\end{subequations}
and satisfy the fundamental relations of the affine Weyl group of type $A_1^{(1)}$:
\begin{equation}
 {w_0}^2={w_1}^2=(w_0w_1)^\infty=1,
\end{equation}
we can set $W(A_{1,|\beta|^2=14}^{(1)})=\langle w_0,w_1\rangle$.
Note that the action of $W(A_1^{(1)})=\langle s_0,s_1\rangle$ and 
that of $W(A_{1,|\beta|^2=14}^{(1)})=\langle w_0,w_1\rangle$ commute.
Moreover, since the action of $\pi$ on the roots are given by
\begin{equation}
 \pi:(\alpha_0,\alpha_1,\beta_0,\beta_1)\mapsto(\alpha_1,\alpha_0,\beta_1,\beta_0),
\end{equation}
we can set Aut$((A_1+A_{1,|\beta|^2=14})^{(1)})=\langle \pi\rangle$.
Note that the transformation $\pi$ satisfies the following relations
\begin{equation}
 \pi^2=1,\quad \pi s_0=s_1 \pi,\quad \pi w_0=w_1 \pi.
\end{equation}
Therefore, the group of Cremona isometries 
\begin{equation}
 W(A_1^{(1)})\rtimes {\rm Aut}(A_6^{(1)})
 =W((A_1+A_{1,|\beta|^2=14})^{(1)})\rtimes {\rm Aut}((A_1+A_{1,|\beta|^2=14})^{(1)})
\end{equation}
form the extended affine Weyl group of type $(A_1+A_1)^{(1)}$, denoted by $\widetilde{W}((A_1+A_1')^{(1)})$.

Finally, we construct the $\tau$ functions after 
\cite{TsudaT2008:MR2425662,TsudaT2006:MR2247459,TT2009:MR2511044,KMNOY2003:MR1984002}. 
We define the variables $f_u$, $f_d$, $g_u$ and $g_d$ by
\begin{equation}
 f=\cfrac{f_u}{f_d},\quad
 g=\cfrac{g_u}{g_d},
\end{equation}
and their polynomial $F_\Lambda$ by
\begin{equation}
 F_\Lambda=F_\Lambda(f_u,f_d,g_u,g_d),
\end{equation}
where $\Lambda=mH_f+nH_g-\sum_{i=1}^8\mu_iE_i$
which corresponds to a curve of bi-degree $(m,n)$ on $\mathbb{P}^1\times\mathbb{P}^1$
passing through base points $P_i$ with multiplicity $\mu_i$.
For example,
\begin{equation}
 F_{H_f+H_g-E_2-E_5-E_8}=\gamma(a_0 f_ug_d+f_dg_u+f_dg_d),
\end{equation}
where $\gamma$ is an arbitrary non-zero complex parameter.
\begin{definition}
We define a mapping $\tau$ on the set
\begin{equation}
 M=\set{w(E_i)}{w\in\widetilde{W}((A_1+A_1')^{(1)}),~ i=1,\dots,8}
\end{equation}
by the following conditions:
\begin{description}
\item[(i)]
\begin{equation}
 w.\tau(\Lambda)=\tau(w^{-1}(\Lambda)),
\end{equation}
where $w\in\widetilde{W}((A_1+A_1')^{(1)})$ and $\Lambda\in M$;
\item[(ii)]
\begin{equation}
 \tau(\Lambda)=\cfrac{F_\Lambda(f_u,f_d,g_u,g_d)}{\tau(E_1)^{\mu_1}\cdots\tau(E_8)^{\mu_8}},
\end{equation}
for $\Lambda=mH_f+nH_g-\sum_{i=1}^8\mu_iE_i\in M$;
\item[(iii)]
\begin{equation}
 \cfrac{F_\Lambda(f_u,f_d,g_u,g_d)}{F_\Lambda(1,1,1,1)}=\tau(E_1)^{\mu_1}\cdots\tau(E_8)^{\mu_8},
\end{equation}
for $\Lambda=mH_f+nH_g-\sum_{i=1}^8\mu_iE_i\in\{D_1,D_3,D_5,D_7\}$,
which are equivalent to
\begin{equation}
 f_u=\tau(E_2)\tau(E_4),\quad
 f_d=\tau(E_3)\tau(E_5),\quad
 g_u=\tau(E_1)\tau(E_3),\quad
 g_d=\tau(E_4)\tau(E_5).
\end{equation}
\end{description}
\end{definition}

Finally, setting
\begin{subequations}
\begin{align}
 &\tau_{-3}=\tau(E_1),\quad
 \tau_{-2}=\tau(E_4)=\tau(E_7),\quad
 \tau_{-1}=\tau(E_5)=\tau(E_8),\\
 &\tau_0=\tau(E_3)=\tau(E_6),\quad
 \tau_1=\tau(E_2)
\end{align}
\end{subequations}
and normalizing the polynomials $F_\Lambda$ to be designed to hold the fundamental relations \eqref{eqns:A1A1_fundamental},
we have completed the proof of Lemma \ref{lemma:tau_A6}.
\def\cprime{$'$} \def\cprime{$'$}


\begin{thebibliography}{10}

\bibitem{ABS2003:MR1962121}
V.~E. Adler, A.~I. Bobenko, and Y.~B. Suris.
\newblock Classification of integrable equations on quad-graphs. {T}he
  consistency approach.
\newblock {\em Comm. Math. Phys.}, 233(3):513--543, 2003.

\bibitem{ABS2009:MR2503862}
V.~E. Adler, A.~I. Bobenko, and Y.~B. Suris.
\newblock Discrete nonlinear hyperbolic equations: classification of integrable
  cases.
\newblock {\em Funktsional. Anal. i Prilozhen.}, 43(1):3--21, 2009.

\bibitem{BollR2011:MR2846098}
R.~Boll.
\newblock Classification of 3{D} consistent quad-equations.
\newblock {\em J. Nonlinear Math. Phys.}, 18(3):337--365, 2011.

\bibitem{BollR2012:MR3010833}
R.~Boll.
\newblock Corrigendum: {C}lassification of 3{D} consistent quad-equations.
\newblock {\em J. Nonlinear Math. Phys.}, 19(4):1292001, 3, 2012.

\bibitem{BollR:thesis}
R.~Boll.
\newblock Classification and {L}agrangian {S}tructure of 3{D} {C}onsistent
  {Q}uad-{E}quations.
\newblock {\em Doctoral Thesis, Technische Universit\"at Berlin}, submitted
  August 2012.

\bibitem{DO1988:MR1007155}
I.~Dolgachev and D.~Ortland.
\newblock Point sets in projective spaces and theta functions.
\newblock {\em Ast\'erisque}, (165):210 pp. (1989), 1988.

\bibitem{FJN2008:MR2425981}
C.~M. Field, N.~Joshi, and F.~W. Nijhoff.
\newblock {$q$}-difference equations of {K}d{V} type and {C}hazy-type
  second-degree difference equations.
\newblock {\em J. Phys. A}, 41(33):332005, 13, 2008.

\bibitem{GRSWC2005:MR2117991}
B.~Grammaticos, A.~Ramani, J.~Satsuma, R.~Willox, and A.~S. Carstea.
\newblock Reductions of integrable lattices.
\newblock {\em J. Nonlinear Math. Phys.}, 12(suppl. 1):363--371, 2005.

\bibitem{HHJN2007:MR2303490}
M.~Hay, J.~Hietarinta, N.~Joshi, and F.~Nijhoff.
\newblock A {L}ax pair for a lattice modified {K}d{V} equation, reductions to
  {$q$}-{P}ainlev\'e equations and associated {L}ax pairs.
\newblock {\em J. Phys. A}, 40(2):F61--F73, 2007.

\bibitem{HHNS2015:MR3317164}
M.~Hay, P.~Howes, N.~Nakazono, and Y.~Shi.
\newblock A systematic approach to reductions of type-{Q} {ABS} equations.
\newblock {\em J. Phys. A}, 48(9):095201, 24, 2015.

\bibitem{HKM2011:MR2788707}
M.~Hay, K.~Kajiwara, and T.~Masuda.
\newblock Bilinearization and special solutions to the discrete {S}chwarzian
  {K}d{V} equation.
\newblock {\em J. Math-for-Ind.}, 3A:53--62, 2011.

\bibitem{HirotaR1977:MR0460934}
R.~Hirota.
\newblock Nonlinear partial difference equations. {I}. {A} difference analogue
  of the {K}orteweg-de {V}ries equation.
\newblock {\em J. Phys. Soc. Japan}, 43(4):1424--1433, 1977.

\bibitem{JNS:paper3}
N.~Joshi and N.~Nakazono.
\newblock Lax pairs of discrete {P}ainlev\'e equations: $({A}_2+{A}_1)^{(1)}$
  case.
\newblock {\em arXiv:1503.04515}.

\bibitem{JNS2014:MR3291391}
N.~Joshi, N.~Nakazono, and Y.~Shi.
\newblock Geometric reductions of {ABS} equations on an {$n$}-cube to discrete
  {P}ainlev\'e systems.
\newblock {\em J. Phys. A}, 47(50):505201, 16, 2014.

\bibitem{JNS:paper4}
N.~Joshi, N.~Nakazono, and Y.~Shi.
\newblock In preparation.

\bibitem{KMNOY2003:MR1984002}
K.~Kajiwara, T.~Masuda, M.~Noumi, Y.~Ohta, and Y.~Yamada.
\newblock {${}_{10}E_9$} solution to the elliptic {P}ainlev\'e equation.
\newblock {\em J. Phys. A}, 36(17):L263--L272, 2003.

\bibitem{KN2015:MR3340349}
K.~Kajiwara and N.~Nakazono.
\newblock Hypergeometric solutions to the symmetric {$q$}-{P}ainlev\'e
  equations.
\newblock {\em Int. Math. Res. Not. IMRN}, (4):1101--1140, 2015.

\bibitem{KNT2011:MR2773334}
K.~Kajiwara, N.~Nakazono, and T.~Tsuda.
\newblock Projective reduction of the discrete {P}ainlev\'e system of type
  {$(A_2+A_1)^{(1)}$}.
\newblock {\em Int. Math. Res. Not.}, (4):930--966, 2011.

\bibitem{KNY2001:MR1876614}
K.~Kajiwara, M.~Noumi, and Y.~Yamada.
\newblock A study on the fourth {$q$}-{P}ainlev\'e equation.
\newblock {\em J. Phys. A}, 34(41):8563--8581, 2001.

\bibitem{KTGR2000:MR1789477}
M.~D. Kruskal, K.~M. Tamizhmani, B.~Grammaticos, and A.~Ramani.
\newblock Asymmetric discrete {P}ainlev\'e equations.
\newblock {\em Regul. Chaotic Dyn.}, 5(3):273--280, 2000.

\bibitem{LooijengaE1981:MR632841}
E.~Looijenga.
\newblock Rational surfaces with an anticanonical cycle.
\newblock {\em Ann. of Math. (2)}, 114(2):267--322, 1981.

\bibitem{NC1995:MR1329559}
F.~Nijhoff and H.~Capel.
\newblock The discrete {K}orteweg-de {V}ries equation.
\newblock {\em Acta Appl. Math.}, 39(1-3):133--158, 1995.
\newblock KdV '95 (Amsterdam, 1995).

\bibitem{NCWQ1984:MR763123}
F.~W. Nijhoff, H.~W. Capel, G.~L. Wiersma, and G.~R.~W. Quispel.
\newblock B\"acklund transformations and three-dimensional lattice equations.
\newblock {\em Phys. Lett. A}, 105(6):267--272, 1984.

\bibitem{NP1991:MR1098879}
F.~W. Nijhoff and V.~G. Papageorgiou.
\newblock Similarity reductions of integrable lattices and discrete analogues
  of the {P}ainlev\'e {${\rm II}$} equation.
\newblock {\em Phys. Lett. A}, 153(6-7):337--344, 1991.

\bibitem{NQC1983:MR719638}
F.~W. Nijhoff, G.~R.~W. Quispel, and H.~W. Capel.
\newblock Direct linearization of nonlinear difference-difference equations.
\newblock {\em Phys. Lett. A}, 97(4):125--128, 1983.

\bibitem{OrmerodCM2012:MR2997166}
C.~M. Ormerod.
\newblock Reductions of lattice m{K}d{V} to {$q$}-{${\rm P}_{\rm VI}$}.
\newblock {\em Phys. Lett. A}, 376(45):2855--2859, 2012.

\bibitem{OrmerodCM:2014arXiv1308.4233}
C.~M. Ormerod.
\newblock Symmetries and {S}pecial {S}olutions of {R}eductions of the {L}attice
  {P}otential {KdV} {E}quation.
\newblock {\em Symmetry Integrability Geom. Methods Appl.}, 10:02--19, 2014.

\bibitem{RG1996:MR1399286}
A.~Ramani and B.~Grammaticos.
\newblock Discrete {P}ainlev\'e equations: coalescences, limits and
  degeneracies.
\newblock {\em Phys. A}, 228(1-4):160--171, 1996.

\bibitem{RGTT2001:MR1838017}
A.~Ramani, B.~Grammaticos, T.~Tamizhmani, and K.~M. Tamizhmani.
\newblock Special function solutions of the discrete {P}ainlev\'e equations.
\newblock {\em Comput. Math. Appl.}, 42(3-5):603--614, 2001.
\newblock Advances in difference equations, III.

\bibitem{SakaiH2001:MR1882403}
H.~Sakai.
\newblock Rational surfaces associated with affine root systems and geometry of
  the {P}ainlev\'e equations.
\newblock {\em Comm. Math. Phys.}, 220(1):165--229, 2001.

\bibitem{TsudaT2006:MR2207047}
T.~Tsuda.
\newblock Tau functions of {$q$}-{P}ainlev\'e {III} and {IV} equations.
\newblock {\em Lett. Math. Phys.}, 75(1):39--47, 2006.

\bibitem{TsudaT2006:MR2247459}
T.~Tsuda.
\newblock Tropical {W}eyl group action via point configurations and
  {$\tau$}-functions of the {$q$}-{P}ainlev\'e equations.
\newblock {\em Lett. Math. Phys.}, 77(1):21--30, 2006.

\bibitem{TsudaT2008:MR2425662}
T.~Tsuda.
\newblock A geometric approach to tau-functions of difference {P}ainlev\'e
  equations.
\newblock {\em Lett. Math. Phys.}, 85(1):65--78, 2008.

\bibitem{TT2009:MR2511044}
T.~Tsuda and T.~Takenawa.
\newblock Tropical representation of {W}eyl groups associated with certain
  rational varieties.
\newblock {\em Adv. Math.}, 221(3):936--954, 2009.

\bibitem{YamadaY:A1A1tau}
Y.~Yamada.
\newblock Tau functions of ${A}_6^{(1)}$-surface $q$-{P}ainlev\'e systems.
\newblock {\em Private communication}.

\end{thebibliography}
\end{document}